\documentclass[10pt, journal]{IEEEtran}

\usepackage{cite}
\ifCLASSINFOpdf
   \usepackage[pdftex]{graphicx}
\else
   \usepackage[dvips]{graphicx}
\fi
\usepackage{amsmath}
\allowdisplaybreaks[4]
\usepackage{algorithmic}
\usepackage{algorithm}

\usepackage{amsmath,amsfonts,amsthm,bm} 

\usepackage{array}

\ifCLASSOPTIONcompsoc
  \usepackage[caption=false,font=normalsize,labelfont=sf,textfont=sf]{subfig}
\else
  \usepackage[caption=false,font=footnotesize]{subfig}
\fi
\usepackage{url}

\usepackage{amsthm}

\newtheorem{definition}{\textbf{Definition}}

\usepackage{amssymb}
\usepackage{mathtools}
\usepackage{cases}
\usepackage{autobreak}

\usepackage{booktabs}
\usepackage{multirow}

\usepackage{color}
\usepackage{hyperref}
\usepackage{cite}

\usepackage{amsmath}
\usepackage{extarrows}

\usepackage{comment}

\newcommand{\E}{\mathbb{E}}
\newcommand{\diag}{\mathrm{diag}}
\newcommand{\Prob}{\mathbb{P}}

\newcommand{\Var}{\mathrm{Var}}

\DeclareMathOperator{\Tr}{Tr}

\newcommand{\FR}{{\bold{R}}}
\newcommand{\FT}{{\bold{T}}}
\newcommand{\FS}{{\bold{S}}}

\newcommand{\RNum}[1]{\uppercase\expandafter{\romannumeral #1\relax}}

\newtheorem{remark}{Remark}
\newtheorem{theorem}{Theorem}

\newtheorem{lemma}{Lemma}
\newtheorem{proposition}{Proposition}

\begin{document}
%
\title{Outage Probability and Finite-SNR DMT Analysis for IRS-aided MIMO Systems: How Large IRSs Need to Be?}
%
%
%

\author{Xin Zhang,~\IEEEmembership{Graduate Student Member,~IEEE}, Xianghao Yu,~\IEEEmembership{Member,~IEEE}, and S.H. Song,~\IEEEmembership{Senior Member,~IEEE}
\thanks{
The authors are with the Department of Electronic and Computer Engineering, The Hong Kong University of Science
and Technology, Hong Kong (e-mail: xzhangfe@connect.ust.hk, \{eexyu, eeshsong\}@ust.hk).
}
}
\maketitle

\begin{abstract}
Intelligent reflecting surfaces (IRSs) are promising enablers for high-capacity wireless communication systems by constructing favorable channels between the transmitter and receiver. However, general, accurate, and tractable outage probability analysis for IRS-aided multiple-input-multiple-output (MIMO) systems is not available in the literature. In this paper, we first characterize the distribution of the mutual information (MI) for IRS-aided MIMO systems by capitalizing on large random matrix theory (RMT). Based on this result, a closed-form approximation for the outage probability is derived and a gradient-based algorithm is proposed to minimize the outage probability with statistical channel state information (CSI). We also investigate the diversity-multiplexing tradeoff (DMT) with finite signal-to-noise ratio (SNR). Based on these theoretical results, we further study the impact of the IRS size on system performance. In the high SNR regime, we provide closed-form expressions for the ergodic mutual information (EMI) and outage probability as a function of the IRS size, which analytically reveal that the benefit of increasing the IRS size saturates quickly. Simulation results validate the accuracy of the theoretical analysis and confirm the increasing cost to improve system performance by deploying larger IRSs. For example, for an IRS-aided MIMO system with 20 antennas at both the transmitter and receiver, we need to double the size of the IRS to increase the throughput from 90\% to 95\% of its maximum value. 
\end{abstract}


%
\IEEEpeerreviewmaketitle

\section{Introduction}
%
%
%
%
Intelligent reflecting surfaces (IRSs) have attracted extensive interests from both academia and industry and are considered as one of the promising solutions for future high-capacity communication systems~\cite{wu2019towards}. By designing the controllable phase shifts, IRSs can customize favorable channels between the transceivers~\cite{yu2021smart}, thus increasing the throughput and the reliability of wireless links~\cite{matthiesen2020intelligent}. In addition, IRSs are energy-efficient due to their passive nature. Inspired by these advantages, IRSs have been applied to various wireless systems such as massive multiple-input-multiple-output (MIMO)~\cite{wang2021massive}, simultaneous wireless information and power transfer (SWIPT)~\cite{xu2021optimal}, and non-orthogonal multiple access (NOMA)~\cite{cheng2021downlink} systems. 


There have been works investigating the design and performance analysis of IRS-aided systems~\cite{mei2021performance},~\cite{zhang2019analysis},~\cite{han2019large}. To this end, capacity (throughput) and outage probability (reliability) are two important performance measures. The ergodic mutual information (EMI) (average throughput) has been well studied for single-input-single-output (SISO)~\cite{van2020coverage} and multiple-input-single-output (MISO) systems~\cite{han2019large}. The EMI of IRS-aided MIMO systems over Rician channels was investigated by random matrix theory (RMT) with an accurate and closed-form approximation~\cite{zhang2021large}. 

The outage probability of IRS-aided SISO systems was evaluated in previous works. In~\cite{cheng2021downlink}, a closed-form expression for the outage probability of NOMA SISO networks was obtained under Nakagami fading using central limit theorem (CLT).
The outage probability of a SISO system with multiple IRSs over Rician channels was analyzed in~\cite{zhang2019analysis}. In~\cite{wang2020outage}, the outage probability of IRS-aided vehicular communication systems was evaluated using CLT. Atapattu \emph{et al.} derived a closed-form expression for the outage probability and provided the optimal phase shifts design~\cite{atapattu2020reconfigurable}. It was also shown that the decreasing rate of the outage probability is related to the size of IRSs.



The outage probability of IRS-aided MISO systems was also studied. In~\cite{lin2020reconfigurable}, considering reflection pattern modulation (RPM), a closed-form approximation for the asymptotic outage probability over Rician channels was obtained using Gamma approximation. In~\cite{guo2020outage}, with maximum-ratio transmission (MRT), the expression of the outage probability and its asymptotically-optimal form were given, while the phase shifts were optimized to minimize the outage probability. In~\cite{zhou2020framework} and~\cite{hong2020robust}, the robust design of IRS-aided MISO systems was investigated with imperfect channel state information (CSI). In~\cite{bereyhi2022channel}, the conventional CLT was used to obtain a closed-form expression of the outage probability.


The analysis in SISO or MISO systems utilized variable and vector based methods (conventional CLT), which are not applicable to MIMO systems. In fact, characterizing the cascaded channel of IRS-aided MIMO system involves the investigation of the spectral distribution of the product of two random matrices, which is a challenging RMT problem~\cite{zheng2016asymptotic}. As far as the authors know, the outage probability of IRS-aided \textit{MIMO systems} was only investigated in~\cite{shi2021outage}, where Mellin transform and finite-regime RMT were leveraged to derive the outage probability over Rayleigh fading channels with channel correlation at one side of the transceivers. 
In other words, a generic and tractable outage probability characterization for IRS-aided MIMO systems is not available in the literature.

In this paper, we first characterize the distribution of the MI for IRS-aided MIMO systems and utilize it to evaluate the outage probability. We then propose a gradient descent algorithm to minimize the outage probability by optimizing the phase shifts. The results about the outage probability are then utilized to investigate the diversity-multiplexing tradeoff (DMT). The SNR-asymptotic DMT was proposed in~\cite{zheng2003diversity} to characterize the trade-off between diversity gain (reliability) and multiplexing gain (spectral efficiency)~\cite{zheng2016asymptotic}, which however is not accurate in the finite SNR regime. In this paper, we will investigate the finite-SNR DMT~\cite{narasimhan2006finite}~\cite{loyka2010finite} of IRS-aided MIMO systems. Finally, the impact of the IRS size on the system performance is studied to answer the question: How large IRSs need to be? 
\subsection{Contributions}
The main contributions of this paper are listed as follows.
\begin{itemize}
  \item[1)] The distribution of the MI for IRS-aided MIMO systems is first derived. Based on the result, an approximation on the outage probability over general correlated channels is obtained with only statistical CSI. To the best of the authors' knowledge, this is the first analytical result regarding the outage probability in IRS-aided MIMO systems over general correlated channels. Numerical results validate the accuracy of the proposed method. 
    \item[2)] With the derived outage probability, a gradient algorithm is proposed to minimize the outage probability by optimizing the phase shifts at the IRS, assuming only statistical CSI. Numerical results show that the algorithm can efficiently decrease the outage probability. 
   \item[3)] A closed-form expression is obtained for the finite-SNR DMT of IRS-aided MIMO systems, which is not available in the literature. An interesting observation is that the finite-SNR DMT is highly related to the ratio between the mean and the standard deviation of the MI. The accuracy of the expression is validated by numerical results. 
    \item[4)] The impact of the size of the IRS on system performance is investigated. To this end, we first propose the concept called \textit{IRS efficiency} to measure the efficiency of increasing the IRS size in achieving the maximum throughput. The expression of the outage probability with respect to the IRS size over uncorrelated channels is explicitly given in the high SNR regime. Based on the theoretical analysis and simulation results, we have two key observations. First, when the size of the IRS is infinitely large, the performance of the two-hop IRS system is the same as that of the single-hop link. Second, the benefit of increasing the size of the IRS saturates quickly. For example, for an IRS-aided system with 20 antennas at the transceivers over independent channels, we need to double the size of the IRS to increase the throughput from 90\% to 95\% of its maximum value. 
\end{itemize}


\subsection{Organizations}
The rest of this paper is organized as follows. Section~\ref{sys_model} presents the system model for the IRS-aided MIMO system. Section~\ref{pre_res} introduces 
the main results including the characterization of the distribution for the MI. The analysis and optimization of the outage probability and the finite-SNR DMT are given in Section~\ref{out_dmt}. Section~\ref{sec_siz} discusses the effect of the IRS size on system performance. The theoretical results are validated in Section~\ref{simu} by numerical simulations. Finally, Section~\ref{sec_con} concludes the paper. 

\textit{Notations}:  Bold upper case letters and bold lower case letters represent the matrix and vector, respectively. $\mathrm{Re} \left\{ \cdot \right\}$ denotes the real part of a complex number. $\mathbb{P}(\cdot)$ is the probability measure. $\mathbb{C}^{N}$ and $\mathbb{C}^{M \times N}$ represent the space of $N$-dimensional vectors and the space of $M$-by-$N$ matrices, respectively. $\bold{A}^{H}$ represents the conjugate transpose of $\bold{A}$. $[\bold{A}]_{i,j}$ represents the $i,j$-th entry of $\bold{A}$. $\otimes$ denotes the element-wise product of matrices. $\Tr\bold{A} $ and $\|\bold{A} \|$ represent the trace and the spectral norm of $\bold{A}$. $\E$ represents the expectation operator. $\Phi(x)$ is the cumulative distribution function (CDF) of standard Gaussian distribution and $Q(\cdot)$ is the $Q$-function, where $Q(x)=1-\phi(x)$. $\mathcal{CN}$ and $\mathcal{N}$ represent the circularly complex Gaussian and real Gaussian distribution, respectively.
$ \xrightarrow[N \rightarrow \infty]{\mathcal{D}} $, $ \xrightarrow[N \rightarrow \infty]{\mathcal{P}} $, and $ \xrightarrow[N \rightarrow \infty]{{a.s.}} $ denote the convergence in distribution, the convergence in probability and the almost sure convergence, respectively. $O(\cdot)$, $o(\cdot)$, and $\Theta(\cdot)$ represent the Big-O, the Little-o, and the Big-Theta notations, respectively. Specifically, $f(n)\in O(g(n))$ if and only if there exists a positive constant $c$ and a nonnegative integer $n_{0}$ such that $f(n) \le cg(n)$ for all $n \ge n_{0}$. $f(n)\in o(g(n))$ if and only if there exists a nonnegative integer $n_{0}$ such that $f(n) \le cg(n)$ for all $n \ge n_{0}$ for all positive $c$. $f(n) \in \Theta(g(n))$ if and only if there exist positive $c_1$ and $c_2$ and nonnegative integer $n_{0}$ such that $c_1 g(n) \le f(n) \le c_2 g(n)$ for all $n\ge n_{0}$~\cite{cormen2009introduction}. 

\section{System Model}
\label{sys_model}
\subsection{System Model}
Consider a point-to-point downlink MIMO communication system, where an IRS is deployed to establish favorable communication links for the user equipment (UE) that would otherwise be blocked. There are $M$ antennas at the basestation (BS) and $N$ antennas at the UE, and the number of elements at the IRS is $L$. Given the line-of-sight (LoS) path is blocked, the received signal $\bold{y}\in \mathbb{C}^{N}$ is given by
\begin{equation}
\label{eq_model}
\bold{y}=\bold{H}_{1}\bold{\Psi}\bold{H}_2\bold{s}+\bold{n},
\end{equation}
where $\bold{s}\in \mathbb{C}^{M}$ denotes the transmitted signal with unit average transmit power, i.e., $\E|s_{i}|^2=1, i=1,2,...,M $, and $\bold{n}\in \mathbb{C}^{N} \sim \mathcal{CN}(0, \sigma^2\bold{I}_{N}) $ represents the additive white Gaussian noise (AWGN) with variance $\sigma^2$. $\bold{H}_{2} \in \mathbb{C}^{L \times M}$ and $\bold{H}_{1} \in \mathbb{C}^{N \times L}$ represent the channel matrices from the BS to the IRS and from the IRS to the UE, respectively. $\bold{\Psi}\in \mathbb{C}^{L\times L}=\diag(\psi_1,\psi_2,...,\psi_L)=\diag(e^{\jmath{\theta_{1}}},e^{\jmath{\theta_{2}}},...,e^{\jmath{\theta_{L}}})$ with $\theta_{i} \in [0, 2\pi),~ i=1,2,...,L$, denotes the phase shifts imposed by the IRS, and we define $\bm{\theta}=(\theta_{1},\theta_{2},...,\theta_{L})$. In this paper, we consider the general Rayleigh model with
\begin{equation}
\label{cha_mod}
\bold{H}_{1}=\bold{R}_{1}^{\frac{1}{2}}\bold{X}\bold{T}^{\frac{1}{2}}_{1},~\bold{H}_{2}=\bold{R}_{2}^{\frac{1}{2}}\bold{Y}\bold{T}^{\frac{1}{2}}_{2},
\end{equation}
where $\bold{R}_{i}$ and $\bold{T}_{i},~i=1,2$, are four positive semi-definite correlation matrices. $\bold{R}_{1}$ and $\bold{T}_2$ denote the correlation matrices of receive and transmit antennas, respectively. $\bold{T}_{1}$ and $\bold{R}_{2}$ denote the transmit and receive correlation matrices of the IRS, respectively. $\bold{X}\in\mathbb{C}^{N \times L}$ and $\bold{Y}\in\mathbb{C}^{L \times M}$ are two independent and identically distributed (i.i.d.) Gaussian random matrices, whose entries follow $\mathcal{CN}(0, \frac{1}{L})$ and $\mathcal{CN}(0, \frac{1}{M})$, respectively. We assume that statistical CSI, i.e., correlation matrices of the channel, is available. To obtain the statistical CSI, the samples of the separate channels are needed. In practice, we can estimate the IRS-UE channel and the BS-IRS channel separately by the methods proposed in~\cite{liu2020matrix},~\cite{hu2021two} and then estimate the corresponding channel covariance matrices based on the techniques proposed in~\cite{liang2001downlink},~\cite{chen2010shrinkage}.

\subsection{Mutual Information and Outage Probability}
The MI of the IRS-aided MIMO system is given by
\begin{equation}
\label{ori_mi}
I(\rho)=\log\det(\bold{I}_{N}+\rho\bold{H}_{1}\bold{\Psi}\bold{H}_2\bold{H}_{2}^{H}\bold{\Psi}^{H}\bold{H}_1^{H}),
\end{equation}
where $\rho=\frac{P}{M\sigma^2}$ with $P$ denoting the total transmit power. Based on the distribution of the MI, we can evaluate the performance of IRS-aided MIMO systems with different metrics. For example, the average throughput can be determined by the EMI as $\E I(\rho)$. On the other hand, the reliability of the system can be measured by the outage probability, which, for a preset transmission rate $R$, can be written as 
\begin{equation}
P_{out}(R)=\Prob(I(\rho)<R).
\end{equation}
The EMI of IRS-aided MIMO systems has been obtained in the literature~\cite{zhang2021large}. In the following, we first investigate the distribution of $I(\rho)$ and then utilize the result to analyze the outage probability and the finite-SNR DMT of IRS-aided MIMO systems. 

\section{Characterization of the MI for IRS-aided MIMO Systems}
\label{pre_res}
Before introducing our main results, we first present some preliminary results including the approximation of the EMI. The analysis is based on RMT, which has been shown to be efficient in analyzing MIMO systems~\cite{couillet2011deterministic},~\cite{wen2012deterministic}.
\subsection{Assumptions and Existing Results on EMI}
The results of this paper are developed based on the following assumptions.

\textbf{Assumption 1.} $0<\lim\inf\limits_{M \ge 1}  \frac{M}{L} \le \frac{M}{L}  \le \lim \sup\limits_{M \ge 1} \frac{M}{L} <\infty$, $0<\lim \inf\limits_{M \ge 1}  \frac{M}{N} \le \frac{M}{N}  \le \lim \sup\limits_{M \ge 1}  \frac{M}{N} <\infty$.

\textbf{Assumption 2.} $\lim \sup\limits_{M\ge 1} \| \bold{R}_{i}\| <\infty$, $\lim \sup\limits_{M\ge 1} \| \bold{T}_{i}\| <\infty$, $i=1,2$~\cite{couillet2011deterministic}~\cite{wen2012deterministic}.

\textbf{Assumption 3.} $\inf\limits_{M\ge 1} \frac{1}{M}\Tr\bold{R}_{1}>0  $, $\inf\limits_{M\ge 1} \frac{1}{M}\Tr\bold{T}_{2}>0 $, $\inf\limits_{M\ge 1} \frac{1}{M}\Tr\bold{T}_{1}\bold{\Psi} \bold{R}_{2}\bold{\Psi}^{H}>0$~\cite{zhang2021bias}~\cite{hachem2008new}.

 \textbf{A.1} is the asymptotic regime considered for the large-scale system, where the dimensions of the system ($M$, $N$, and $L$) grow to infinity at the same paces. \textbf{A.2} and~\textbf{A.3} restrict the rank of the correlation matrices so that the extremely low-rank case, i.e., the ranks of the correlation matrices do not increase with the number of antennas, will not occur.

Given the eigenvalue decompositions $\bold{R}=\bold{U}_{R}\bold{R}_{1}\bold{U}_{R}^{H}$, $\bold{S}=\bold{U}_{S}\bold{T}^{\frac{1}{2}}_{1}\bold{\Psi}\bold{R}_{2}\bold{\Psi}^{H}\bold{T}^{\frac{1}{2}}_{1}\bold{U}_{S}^{H}$, $\bold{T}=\bold{U}_{T}\bold{T}_{2}\bold{U}_{T}^{H}$, where $\bold{R}$, $\bold{S}$, and $\bold{T}$ are diagonal matrices, and the singular value decomposition (SVD) $\bold{T}^{\frac{1}{2}}_{1}\bold{\Psi}\bold{R}_{2}^{\frac{1}{2}}=\bold{U}_{S}^{H}\bold{S}^{\frac{1}{2}}\bold{V}_{S}$, the MI in~(\ref{ori_mi}) can be written as
\begin{equation}
\begin{aligned}
&I(\rho)\overset{a}{=}\log\det(\bold{I}_{N}+
\rho \bold{R}_{1}^{\frac{1}{2}}\bold{X}\bold{T}^{\frac{1}{2}}_{1}\bold{\Psi} \bold{R}_{2}^{\frac{1}{2}}\bold{Y}\bold{T}_{2}\bold{Y}^{H}\bold{R}_{2}^{\frac{1}{2}}
\\
&
\times \bold{\Psi}^{H}\bold{T}^{\frac{1}{2}}_{1}\bold{X}^{H}\bold{R}_{1}^{\frac{1}{2}})
\overset{b}{=}\log\det({\bold{I}}_{N}+\rho\bold{U}_{R}^{H} \bold{R}^{\frac{1}{2}}\bold{U}_{R}\bold{X} \bold{U}_{S}^{H}
\\
&
\times \bold{S}^{\frac{1}{2}} \bold{V}_{S} \bold{Y}\bold{U}_{T}^{H} \bold{T}\bold{U}_{T}
\bold{Y}^{H}\bold{V}_{S}^{H}\bold{S}^{\frac{1}{2}} \bold{U}_{S} \bold{X}^{H}\bold{U}_{R}^{H} \bold{R}^{\frac{1}{2}}\bold{U}_{R})  
\\
&\overset{c}{=}\log\det({\bold{I}}_{N}+\rho\bold{R}^{\frac{1}{2}}\bold{X}'\bold{S}^{\frac{1}{2}}\bold{Y}'\bold{T}\bold{Y}'^{H}\bold{S}^{\frac{1}{2}}\bold{X}'^{H}\bold{R}^{\frac{1}{2}}),
\end{aligned}
\end{equation}
where step $a$ follows by plugging~(\ref{cha_mod}) into~(\ref{ori_mi}). Step $b$ is obtained by plugging in the eigenvalue decompositions. Step $c$ follows from $\bold{X}'=\bold{U}_{R}\bold{X}\bold{U}_{S}^{H}$, $\bold{Y}'=\bold{V}_{S}\bold{Y}\bold{U}_{T}^{H}$, and $\det(\bold{I}+\bold{A}\bold{B})=\det(\bold{I}+\bold{B}\bold{A})$. Due to the unitary invariant attributes of Gaussian random matrices, a Gaussian matrix $\bold{G}$ is equivalent to $\bold{G}'=\bold{U}\bold{G}\bold{V}$ statistically, where $\bold{U}$ and $\bold{V}$ are any unitary matrices. 
Therefore, $\bold{X}$ ($\bold{Y}$) and $\bold{X}'$ ($\bold{Y}'$) are statistically equivalent so that we will not differentiate them in the following. Thus, the equivalent channel matrix can be given by
\begin{equation}
\label{h_rst}
{\bold{H}}=\bold{R}^{\frac{1}{2}}\bold{X}\bold{S}^{\frac{1}{2}}\bold{Y}\bold{T}^{\frac{1}{2}}.
\end{equation}

The following theorem gives an approximation for the EMI. 
\begin{theorem}
\label{mean_app}
 (\cite[Theorem 2]{hoydis2011asymptotic},~\cite[Theorem 5]{hoydis2011iterative} and~\cite[Corollary 1]{zhang2021large})
With the channel matrix $\bold{H}$ given in~(\ref{h_rst}), if \textbf{A.1} and \textbf{A.2} are satisfied, it holds true that, for general random matrices $\bold{X}$ and $\bold{Y}$,
\begin{equation}
\label{nor_emi}
\frac{1}{N}\E I(\rho) \xrightarrow{N\rightarrow \infty} \frac{1}{N}\overline{I}(\rho).
\end{equation} 
When $\bold{X},\bold{Y}$ are Gaussian random matrices~\cite{zhang2021large}, it further holds true that 
\begin{equation}
\E I(\rho) \xrightarrow{N\rightarrow \infty} \overline{I}(\rho),
\end{equation} 
where $\overline{I}(\rho)$ is given by
\begin{equation}
\label{the_mean}
\begin{aligned}
\overline{I}(\rho)=&\log\det(\bold{I}_{N}+\frac{\rho M g\overline{g}}{L\delta}\bold{R} )+\log\det(\bold{I}_{L}+\delta\overline{g}\bold{S})
\\
&+\log\det(\bold{I}_{M}+g\bold{T})-2Mg\overline{g}.
\end{aligned}
\end{equation}
Here, $(\delta, g, \overline{g})$ is the unique positive solution of the following system of equations
 \begin{equation}
 \label{basic_eq}
 \delta=\frac{1}{L}\Tr\bold{R}\bold{Q}_{R},~g=\frac{1}{M}\Tr \bold{S}\bold{Q}_{S}, ~\overline{g}=\frac{1}{M}\bold{T}\bold{Q}_{T},
 \end{equation}
 where 
  \begin{equation}
  \begin{aligned}
 \bold{Q}_{R}&=\left(\frac{1}{\rho} \bold{I}_{N}+\frac{M g\overline{g}}{L\delta}\bold{R}   \right)^{-1},
 \bold{Q}_{S}=\left(\frac{1}{\delta}\bold{I}_{L}+\overline{g}\bold{S}\right)^{-1},
 \\
 \bold{Q}_{T}&=\left(\bold{I}_{M}+g\bold{T}\right)^{-1}.
\end{aligned}
  \end{equation}
\end{theorem}
\begin{remark}
\label{mean_rem}
This theorem indicates that $\overline{I}(\rho)$ is a good approximation for $\E I(\rho)$. For~(\ref{nor_emi}) to hold, the random matrices $\bold{X}$ and $\bold{Y}$ are not necessarily Gaussian where the EMI is normalized by $N$.~(\ref{the_mean}) indicates that when $\bold{X}$ and $\bold{Y}$ are Gaussian, the convergence holds even if we get rid of the factor $\frac{1}{N}$, which may not hold true for non-Gaussian matrices as discussed in~\cite{hachem2012clt},~\cite{zhang2021bias}. For IRS-aided MIMO systems, the matrix $\bold{S}$ is related to the phase shift matrix $\bold{\Psi}$, which indicates that the solution of~(\ref{basic_eq}) is also related to the phase shifts. The convergence in Theorem~\ref{mean_app} involves the expectation of the MI. On the other hand, by the asymptotic regime \textbf{A.1}, the MI normalized by the number of the receive antennas will converge to a deterministic quantity almost surely when the number of antennas goes to infinity with the same pace~\cite[Theorem 2]{hoydis2011asymptotic}, which implies the occurrence of the channel hardening~\cite{bereyhi2022channel},~\cite{jung2020performance},\cite{zhang2021sum}. 
\end{remark}


Next, we investigate the fluctuation of $I(\rho)$. The challenge arises from the fact that the effective channel is the product of two random matrices. There are some related results in the literature. The CLT of linear spectral statistics for $F$-matrices was given in~\cite{zheng2012central}. In~\cite{gotze2017distribution}, the CLT of linear spectral statistics with general non-Gaussian entries was given for the product of two i.i.d. random matrices, which was also investigated in~\cite{zheng2016asymptotic} by a free probability approach for Gaussian random matrices. The authors of~\cite{zheng2016asymptotic} gave the CLT for the MI of double-Rayleigh channels when $M=N$. However, the CLT for the MI over correlated channels has not been considered and is one of the main contributions of this paper. Based on the CLT, we also investigate the outage probability for IRS-aided MIMO systems using large RMT.

\label{sec_out}
 \begin{table*}[!htbp]
\centering
\caption{List of Expressions.}
\label{var_list}
\begin{tabular}{|cc|cc|cc|cc|}
\toprule
Symbols& Expression &  Symbols & Expression &Symbols& Expression &  Symbols& Expression \\
\midrule
$\gamma_{R}$ & $\frac{1}{L}\Tr\bold{R}^{2}\bold{Q}_{R}^{2}$
&
$\gamma_{R,I}$ & $\frac{1}{L}\Tr\bold{R}\bold{Q}_{R}^{2}$
&
$\gamma_{S}$ &  $\frac{1}{M}\Tr\bold{S}^2\bold{Q}_{S}^2$
&
$\gamma_{S,I}$ &  $\frac{1}{M}\Tr\bold{S}\bold{Q}_{S}^2$
\\
$\gamma_{T}$& $\frac{1}{M}\Tr\bold{T}^2\bold{Q}_{T}^2$
&
$\gamma_{T,I}$& $\frac{1}{M}\Tr\bold{T}\bold{Q}_{T}^2$
&
$\eta_{R}$ & $\frac{1}{L}\Tr\bold{R}^3\bold{Q}_{R}^3$
&
$\eta_{R,I}$ & $\frac{1}{L}\Tr\bold{R}^2\bold{Q}_{R}^3$
\\
$\eta_{S}$  & $\frac{1}{M}\Tr\bold{S}^3\bold{Q}_{S}^{3}$
&
 $\eta_{S,I}$ & $\frac{1}{M}\Tr\bold{S}^2\bold{Q}_{S}^{3}$
 &
  $\eta_{T}$ & $\frac{1}{M}\Tr\bold{T}^3\bold{Q}_{T}^{3}$
  &
    $\eta_{T,I}$ & $\frac{1}{M}\Tr\bold{T}^{2}\bold{Q}_{T}^{3}$
  \\
  $\Delta_{Y}$& $1-\gamma_{S}\gamma_{T}$
&
$\Gamma$ & $\frac{M}{L\delta^2}(\frac{\gamma_{T,I}^2\gamma_{S}}{\Delta_{Y}}+g^2\gamma_{T})$
&
$\Delta_{X}$ & $1-\gamma_{R}\Gamma$
  &
 $g(\bold{F})$& $\frac{1}{M}\Tr\bold{F}\bold{Q}_{S} $ 
 \\
    $\psi_{T}$& $\frac{1}{M}\Tr\bold{T}^2\bold{Q}_{T}^{4}$
   &
   $\Gamma_{L}$&
   $\Gamma-\frac{\gamma_{S}\psi_{T} }{L\delta^2\Delta_{Y}}$
   &
  $\gamma_{S}(\bold{F})$ & $\frac{1}{M}\Tr\bold{S}\bold{Q}_{S}\bold{F} \bold{Q}_{S}$
  &
    $\eta_{S}(\bold{F})$ & $\frac{1}{M}\Tr\bold{S}\bold{Q}_{S}\bold{F} \bold{Q}_{S}\bold{S} \bold{Q}_{S}$
 \\
\bottomrule
\end{tabular}
\end{table*}

\subsection{Asymptotic Gaussianity of the MI}
To characterize the distribution of the MI, we first prove the Gaussianity of the MI.
\begin{theorem} 
\label{clt_mi}
(CLT for the MI) If \textbf{A.1} - \textbf{A.3} are satisfied, it holds true that
\begin{equation}
\frac{I(\rho)-\overline{I}(\rho)}{\sqrt{V(\rho)}} \xrightarrow[N \rightarrow \infty]{\mathcal{D}}  \mathcal{N}(0,1),
\end{equation}
where $\overline{I}(\rho)$ is given in~(\ref{the_mean}). The asymptotic variance $V(\rho)$ is given by
\begin{equation}
\label{low_var}
\begin{aligned}
V(\rho)&=-\log(1-\gamma_{R}\Gamma_{L})-\log(1-\gamma_{S}\gamma_{T}),
\end{aligned}
\end{equation}
where $\gamma_{R}$, $\gamma_{S}$, $\gamma_{T}$, and $\Gamma_{L}$ are listed in Table~\ref{var_list}.
\end{theorem}
\begin{proof}
The proof of Theorem~\ref{clt_mi} is given in Appendix~\ref{proof_clt}.
\end{proof}
\begin{remark}
Theorem~\ref{clt_mi} indicates the asymptotic Gaussianity of the MI. Note that the variance consists of two terms caused by $\bold{X}$ and $\bold{Y}$, respectively. Different from the \textit{entry}-based CLT used in SISO systems~\cite{cheng2021downlink},~\cite{wang2020outage}, this CLT is developed based on RMT. From the proof of Theorem~\ref{clt_mi}, we can observe that the term $-\frac{\gamma_{S}\psi_{T} }{L\delta^2\Delta_{Y}}$ in $\Gamma_{L}$ can be omitted if $L$ is large, where $\psi_{T}$ is given in Table~\ref{var_list}. In this case, we can use $\Gamma$, listed in Table~\ref{var_list}, to replace $\Gamma_{L}$ and the variance is given by
\begin{equation}
\begin{aligned}
V(\rho)&=-\log(1-\gamma_{R}\Gamma)-\log(1-\gamma_{S}\gamma_{T}).
\end{aligned}
\end{equation}
The large $L$ case will be used in the following analysis as we will investigate the asymptotic performance when $L$ goes to infinity.
\end{remark}

\subsection{Uncorrelated Cases}
To better illustrate the impact of the size of the IRS, we further consider the special case where the channels are i.i.d., i.e., $\bold{R}=\bold{I}_{N}, \bold{S}=\bold{I}_{L}, \bold{T}=\bold{I}_{N}$, which is referred to as the double-Rayleigh model~\cite{zheng2016asymptotic}. To derive the CLT, we first determine the EMI.
\begin{proposition}
\label{mean_pro}
  If \textbf{A.1} holds true, the EMI of IRS-aided MIMO systems
over independent channels with $N=M$ is given by
 \begin{equation}
 \begin{aligned}
  \overline{I}(\rho)
  &\!=\!2N\log(1\!+\!g)\!+\! \frac{N}{\tau}\log(1\!+\! \frac{\tau \rho}{(1+g)^2}) 
   \!-\! \frac{2N g}{1+g},
\end{aligned}
 \end{equation}
 where $\tau=\frac{M}{L}$, and $g$ is determined by the following cubic equation
  \begin{equation}
  \label{cubic_g}
  g^{3}+2g^{2}+(1+\rho \tau - \rho )g-\rho=0,
  \end{equation}
  with $g>0$ and $1+(1-\tau)g>0$.
\end{proposition}
\begin{proof}
Proposition~\ref{mean_pro} can be obtained from Theorem~\ref{mean_app} by setting $\bold{R}=\bold{I}_{N}, \bold{S}=\bold{I}_{L}, \bold{T}=\bold{I}_{N}$, so we omit the proof.
\end{proof}
\begin{remark}
In this case, the parameters in Theorem~\ref{mean_app} degenerate to one parameter $g$, which is described by a cubic equation. Thus, the performance only depends on the ratio between the number of transmit antennas ($M/N$) and the size of the IRS (L), i.e., $\tau$.
\end{remark}

\begin{proposition}
\label{clt_ray}
If \textbf{A.1} holds true, $N=M$, and $L$ is comparable with $M$, the CLT for the MI of IRS-aided MIMO systems over independent channels is given by
\begin{equation}
\frac{I(\rho)-\overline{I}(\rho)}{\sqrt{V(\rho)}} \xrightarrow[N \rightarrow \infty]{\mathcal{D}} \mathcal{N}(0,1),
\end{equation}
where
 \begin{equation}
 \label{ref_rek}
 V(\rho)=\log[\rho (1+g)^2]-\log(\rho+2g^3+2g^2),
 \end{equation}
 and $g$ is identical to that in Proposition~\ref{mean_pro}.
\end{proposition}
\begin{proof}
The proof of Proposition~\ref{clt_ray} is given in Appendix~\ref{proof_ray}.
\end{proof}
\begin{remark}
Proposition~\ref{clt_ray} was also obtained in~\cite[Proposition 3]{zheng2016asymptotic} by a free probability approach while in this paper it is shown as a special case of Theorem~\ref{clt_mi}. The different signs of $g^2$ and $\rho$ in~(\ref{ref_rek}) originate from the fact that the signs of Cauchy transform and Stieljies transform are opposite. 
\end{remark}
 Proposition~\ref{mean_pro} and Proposition~\ref{clt_ray} indicate that the mean and variance of the MI are determined by the root $g$ of the cubic equation~(\ref{cubic_g}), whose coefficients are related to $\tau$ and SNR $\rho$.

\section{Outage Probability and Finite-SNR DMT Analysis}
\label{out_dmt}
\subsection{Outage Probability of General Correlated IRS-aided  MIMO Systems}
As a direct result of Theorem~\ref{clt_mi}, the approximation of the outage probability is given by the following theorem. 
\begin{theorem} 
\label{the_out}
(A closed-form approximation for the outage probability of IRS-aided MIMO systems) Given a preset transmission rate $R$, the outage probability can be 
approximated by
\begin{equation}
\label{app_out}
P_{out}(R)\approx \Phi\left(\frac{R-\overline{I}(\rho)}{\sqrt{V(\rho)}}\right).
\end{equation}
On the other hand, given the outage probability $p_{out}$, the outage rate can be approximated by $R\approx \overline{I}(\rho)+  \sqrt{V(\rho)}\Phi^{-1}(p_{out})$.
\end{theorem}

The above result is different from~\cite{shi2021outage} in the sense that both $\bold{R}_{1}$ and $\bold{T}_{2}$ are not restricted to be an identity matrix, indicating that the proposed method can handle the general correlated channels. In fact, the result in Theorem~\ref{the_out} is also applicable for double-scattering channels~\cite{hoydis2011asymptotic}~\cite{kammoun2019asymptotic}.

\subsection{Outage Probability Optimization}
\label{sec_opt}



Given statistical CSI, the optimization problem for the outage probability can be formulated as
 \begin{equation}
  \begin{aligned}
\mathcal{P}1:~&\min_{\bold{\Psi}} P_{out}(R),~s.t.
\\
& \bold{\Psi}=\diag\left(\psi_1,\psi_2,...,\psi_L \right),
\\
&|\psi_l|=1, l=1,2,...L, 
 \end{aligned}
 \end{equation}
 where $\psi_{l}=\exp(\jmath \theta_l)$.
The challenges arise from two aspects: (1) Evaluation of the objective function; (2) The non-convexity of the problem caused by the uni-modular constraints of $\psi_{i}$. The first one can be resolved by approximating the outage probability using~(\ref{app_out}) in Theorem~\ref{the_out}. Given the SNR $\rho$, the mean $\overline{I}(\rho)$ and variance $V(\rho)$ are functions of the phase shifts $\psi_{l}, l=1,2,...,L$, so we use the notations $\overline{I}(\bold{\Psi})$ and $V(\bold{\Psi})$ instead. Thus, $\mathcal{P}1$ can be rewritten as 
 \begin{equation}
  \begin{aligned}
  \label{p2}
\mathcal{P}2:~&\min_{\bold{\Psi}}~  G(\bold{\Psi})=\Phi\left(\frac{R-\overline{I}(\bold{\Psi})}{\sqrt{V(\bold{\Psi})}}\right),~s.t.
\\
&\bold{\Psi}=\diag\left(\psi_1,\psi_2,...,\psi_L \right),
\\
&|\psi_l |=1,~l=1,2,..., L.
 \end{aligned}
 \end{equation}
Although the parameters $\delta$, $g$, and $\overline{g}$ are coupled with $\theta_l$ as explained in Remark~\ref{mean_rem}, the partial derivatives of the objective function with respect to $\theta_{l}$ can be given in a closed form. Therefore, $\mathcal{P}2$ can be solved using the gradient descent method. Specifically, in each iteration, the update of $\theta_l$ is obtained by searching in the negative gradient direction, until the value of the objective function converges to a stationary point. 

Next, we compute the partial derivatives with respect to $\theta_{l}$, $l=1,2,...,L$, and we use the notation $(\cdot)'_{l}=\frac{\partial (\cdot)}{\partial \theta_{l}}$ to represent the partial derivatives. By the chain rule, the partial derivative of $G(\bold{\Psi})$ with respect to $\theta_{l}$ is given by
\normalsize
\begin{equation}
\label{de1}
\begin{aligned}
G_{l}'(\bold{\Psi})&=\frac{\exp\left(-\frac{ T^2(\bold{\Psi})}{2} \right)T_{l}'(\bold{\Psi})}{\sqrt{2\pi}},
\end{aligned}
\end{equation}
where 
\begin{equation}
T(\bold{\Psi})=\frac{R-\overline{I}(\bold{\Psi})}{\sqrt{V(\bold{\Psi})}},
\end{equation}
\begin{equation}
\begin{aligned}
\nonumber
 T_{l}'(\bold{\Psi})&=\frac{ -\overline{I}_{l}'(\bold{\Psi}) V(\bold{\Psi})-\frac{1}{2}(R- \overline{I}(\bold{\Psi}))V'_{l}(\bold{\Psi})}{V(\rho)^{\frac{3}{2}}},
\\
\overline{I}_{l}'(\bold{\Psi})&=\overline{g}\Tr[(\frac{1}{\delta} \bold{I}_{L}+\overline{g}\bold{T}^{\frac{1}{2}}_{1}\bold{\Psi}\bold{R}_{2}\bold{\Psi}^{H}\bold{T}^{\frac{1}{2}}_{1})^{-1}\bold{F}_{l}].
\\
\end{aligned}
\end{equation}
$\bold{F}_{l}$ is defined as
\begin{equation}
\bold{F}_{l}=\bold{T}_1^{\frac{1}{2}}(\bold{G}_{l}\otimes\bold{R}_{2})\bold{T}_1^{\frac{1}{2}},\quad l=1,2,...,L,
\end{equation}
where
\begin{equation}
\left[\bold{G}_{l}\right]_{p,q}=\left\{
\begin{aligned}
& \jmath e^{\jmath (\theta_{l}-\theta_{q})} ,&p = l, \\
& -\jmath e^{\jmath (\theta_{p}-\theta_{l})} ,&q = l, \\
&0,  &otherwise.
\end{aligned}
\right.
\end{equation}
In fact, $\bold{F}_{l}=(\bold{T}_1^{\frac{1}{2}}\bold{\Psi}\bold{R}_{2}\bold{\Psi}^{H}\bold{T}_1^{\frac{1}{2}})'_{l}$ and the term $V'_{l}(\rho)$ can be given by 
\begin{equation}
V'_{l}(\bold{\Psi})=\frac{\gamma_{S}\gamma_{T,l}'+\gamma_{S,l}'\gamma_{T}}{\Delta_Y}+ \frac{\gamma_{R}\Gamma_{l}'+\gamma_{R,l}'\Gamma}{\Delta_X},
\end{equation}
where
\begin{equation}
\label{gamma_de}
\begin{aligned}
\Gamma_{l}'&=\frac{M}{L\delta^2}[\frac{2\gamma_{T,I}\gamma_{T,I,l}'\gamma_{S}}{\Delta_{Y}} 
+\frac{\gamma_{T,I}^2(\gamma'_{S,l}+\gamma_{S}^{2}\gamma_{T,l}')}{\Delta_{Y}^2} 
\\
+ &2gg_{l}'\gamma_{T}+g^2\gamma_{T,l}' ]-\frac{2M\delta_{l}'}{L\delta^3}(\frac{\gamma_{S}\gamma_{T,I}^2}{\Delta_{Y}} 
+ g^2\gamma_{T}).
\end{aligned}
\end{equation}
The derivatives of $\gamma$'s in~(\ref{gamma_de}) are given as
\begin{equation}
\begin{aligned}
\gamma_{R,l}'
&=\frac{ -2M\eta_{R}(\delta g'_{l}\overline{g} + \delta g\overline{g}'_{l} -g\overline{g}\delta'_{l} )}{L\delta^2},
\\
\gamma_{T,l}'&=-2g'\eta_{T},
\\
\gamma_{S,l}'
&=-2\overline{g}' \eta_{S} -2\overline{g}\eta_{S}(\bold{F}_{l})+\frac{2\delta' \eta_{S,I}}{\delta^2}+2\gamma_{S}(\bold{F}_{l}),
\end{aligned}
\end{equation}
where $(\delta'_{l},g'_{l},\overline{g}'_{l})$ can be computed by the following lemma.
\begin{lemma} 
\label{de_del}
Given that $(\delta, g, \overline{g})$ is the positive solution of~(\ref{basic_eq}) and $\bold{p}_{l}=(\delta'_{l},g'_{l},\overline{g}'_{l})^{T}$, it holds true that
\begin{equation}
\begin{aligned}
\bold{A}\bold{p}_{l}
=\bold{q}_{l},
\end{aligned}
\end{equation}
and $|\bold{A}| >0$, where $\bold{A}$ and $\bold{q}_{l}$ are defined by
\begin{equation}
\label{A_def}
\begin{aligned}
\bold{A}&=\begin{bmatrix}
z\gamma_{R,I}
& \frac{M\overline{g}\gamma_{R}  }{L} 
& \frac{Mg\gamma_{R} }{L}
\\
-\frac{\gamma_{S,I}}{{\delta^2}}
&
1
&      \gamma_S
\\
0 & \gamma_T & 1
\\
\end{bmatrix},
\\
\bold{q}_{l}
&=
\begin{bmatrix}
0 & g(\bold{F}_{l})-2\overline{g}\gamma_{S}(\bold{F}_{l}) & 0
\end{bmatrix}^{T}.
\end{aligned}
\end{equation}
Therefore, we have
\begin{equation}
\bold{p}_{l}
=\bold{A}^{-1}\bold{q}_{l}.
\end{equation}
\begin{proof}
The proof can be obtained by taking derivatives of $(\delta, g, \overline{g})$ with respect to $\theta_{l}$ at both sides of~(\ref{basic_eq}), and the details are omitted here. We have $|\bold{A}|>0$ because
\begin{equation}
| \bold{A} |=z\gamma_{R,I}(1-\gamma_{S}\gamma_{T})+\frac{M\gamma_{R}\gamma_{S,I}\gamma_{T,I}}{L\delta^2},
\end{equation}
and the two terms at the right-hand side (RHS) are positive, so $\bold{A}$ is invertible.
\end{proof}
\end{lemma}

Next, we provide the gradient descent method. We use the Armijo-Goldstein (AG) line search method~\cite{armijo1966minimization} to find an expected decrease of the objective function based on the local gradients. The AG condition is given in the $4$th line of Algorithm~\ref{gra_alg}. The algorithm can be proved to converge to a stationary point by following the proof of Theorem 1 in~\cite{yang2019inexact} or Theorem 1 in~\cite{ma2020low}. The performance of the algorithm will be evaluated in Section~\ref{simu}.

\textbf{Complexity Analysis:} In each iteration, we need to compute $(\delta,g,\overline{g})$ according to the updated $\bold{\Phi}$ by solving~(10). To avoid the computation induced by the matrix inversion, we can use the diagonal matrix $\bold{\Lambda}_{R}=\bold{U}\bold{R}\bold{U}^{H}$ to replace $\bold{R}$ by its eigenvalue decomposition, whose complexity is $O(N^3)$~\cite{pan1999complexity}. Similar operations can be performed on $\bold{S}$ and $\bold{T}$. By the analysis in Appendix~E, the complexity of obtaining an $\varepsilon$-approximation of the solution is $N\log^2(\frac{1}{\varepsilon})$. Assume that the numbers of outer and inner iterations are $N_{outer}$ and $N_{inner}$, respectively. Then the total complexity of the algorithm will be $O(N_{outer}N_{inner}(N^3+N\log^2(\frac{1}{\varepsilon})) )$, where $N_{outer}$ and $N_{inner}$ are determined by the convergence condition and the second-order functional attributes of $G(\bold{\Psi})$. As a smaller value for the objective function is obtained in each iteration, the algorithm will converge to a stationary point. With the statistical CSI, the phase shifts may not need to be updated frequently so that the overall complexity of Algorithm 1 is not high in real systems.

\begin{algorithm} 
\caption{Gradient Descent Algorithm for the Phase Shift Matrix $\bold{\Psi}$} 
\label{gra_alg} 
\begin{algorithmic}[1] 
\REQUIRE  $\bm{\theta}^{\left(0 \right)}$, initial stepsize $\alpha_{0}$, scaling factor $0<c<1$ and control parameter $0<\beta<1$. 
Set $t=0$.
\REPEAT
\STATE Compute the gradient 
$$\nabla_{\bm{\theta}}G(\bold{\Psi})=  (\frac{\partial G(\bold{\Psi})}{\partial \theta_{1}}, \frac{\partial G(\bold{\Psi})}{\partial \theta_{2}},..., 
\frac{\partial G(\bold{\Psi})}{\partial \theta_{L}})^{T},$$ 
according to~(\ref{de1}) and its direction $\bold{d}^{(t)}=\frac{\nabla_{\bm{\theta}}G(\bold{\Psi})}{\|\nabla_{\bm{\theta}}G(\bold{\Psi}) \|}$.
\STATE $\alpha\leftarrow\alpha_{0}$.
\WHILE{$G(\bold{\Psi}^{(t)})-G( \diag[\exp(\jmath\bm{\theta}^{(t)}-\alpha \jmath \bold{d}^{(t)}) ])<\alpha\beta \| \nabla_{\bm{\theta}}G(\bold{\Psi}^{(t)})\|$} 
\STATE  $\alpha \leftarrow c\alpha$.
\ENDWHILE 
	\STATE $\bm{\theta}^{(t+1)} \leftarrow \bm{\bold{\theta}}^{(t)}-\alpha \bold{d}^{(t)}$.
	\STATE $\bold{\Psi}^{(t+1)} \leftarrow \diag[\exp(\jmath\bm{\theta}^{(t+1)})]$.
\STATE $t \leftarrow  t+1$.
\UNTIL Convergence.
\ENSURE  $\bm{\theta}^{(t)}, \bold{\Psi}^{(t)}$.
\end{algorithmic}
\end{algorithm}

\subsection{Finite SNR DMT for IRS-aided Systems}
\label{sec_dmt}
In this section, we investigate the \textit{finite-SNR DMT} of IRS-aided systems. DMT was proposed in~\cite{zheng2003diversity} to characterize the trade-off between diversity and multiplexing gain, which 
 is typically referred to as the \textit{asymptotic-SNR DMT}. 
Here, we investigate the finite-SNR DMT~\cite{narasimhan2006finite},~\cite{loyka2010finite}, which provides more insightful information for  the low and moderate SNR regimes.


By the definition of finite-SNR DMT~\cite{loyka2010finite}, the multiplexing gain is given by 
\begin{equation}
\label{def_mul}
m=\frac{kR}{\E I(\rho)}\approx\frac{kR}{ \overline{I}(\rho)},
\end{equation} 
where  $k=\min(L,M,N)$ and $R$ is the data rate. By approximating the outage probability using the CLT in Theorem~\ref{clt_mi}, we obtain the following theorem regarding the finite-SNR DMT for IRS-aided systems in Rayleigh channels with equal power allocation among different antennas.
\begin{theorem} 
\label{fin_dmt}
The DMT in the finite-SNR regime can be approximated by
\begin{equation}
\label{dmt_exp}
\begin{aligned}
d(m,\rho)
&=\frac{z(m-k)}{d\sqrt{2\pi}}\frac{\exp(-\frac{(m-k)^2H^2(z)}{2k^2})H'(z)}{\Phi(\frac{(m-k)H(z)}{k} )},
\end{aligned}
\end{equation}
where $H(z) =\frac{\overline{I}(z)}{\sqrt{V(z)}}$. $\overline{I}(z)$ is obtained by replacing $\rho$ in $\overline{I}(\rho)$ with $z=\rho^{-1}$, and the same holds for $V(z)$. $H'(z)=\frac{\mathrm{d} H(z)  }{\mathrm{d} z}$ is given as 
\begin{equation}
\label{de_z}
\begin{aligned}
H'(z)=\frac{ \overline{I}'(z) V(z)-\frac{1}{2} \overline{I}(z)V'(z)}{V(z)^{\frac{3}{2}}},
\end{aligned}
\end{equation}
where
\begin{equation}
\label{de_z}
\begin{aligned}
&\overline{I}'(z)=-\frac{N}{z}+ \Tr\bold{Q}_{R}, 
\\
&V'_{l}(\rho)=\frac{\gamma_{S}\gamma_{T}'+\gamma_{S}'\gamma_{T}}{\Delta_Y}+ \frac{\gamma_{R}\Gamma'+\gamma_{R}'\Gamma}{\Delta_X},
\\
&\Gamma'=\frac{M}{L\delta^2}[\frac{2\gamma_{T,I}\gamma_{T,I}'\gamma_{S}}{\Delta_{Y}} 
+\frac{\gamma_{T,I}^2(\gamma'_{S}+\gamma_{S}^{2}\gamma_{T}')}{\Delta_{Y}^2} 
\\
&+ 2gg'\gamma_{T}+g^2\gamma_{T}' ]-\frac{2M\delta'}{L\delta^3}(\frac{\gamma_{S}\gamma_{T,I}^2}{1-\gamma_{S}\gamma_{T}} 
+ g^2\gamma_{T}),
\\
&\gamma_{R}'=\frac{ -2M\eta_{R}(\delta g'\overline{g} + \delta g\overline{g}' -g\overline{g}\delta' )}{L\delta^2}-2\eta_{R,I},
\\
&\gamma_{S}'
=-2\overline{g}' \eta_{S} +\frac{2\delta' \eta_{S,I}}{\delta^2},
\\
&\gamma_{T}'=-2g'\eta_{T},
\\
&[\delta',g',\overline{g}']^{T}=\bold{A}^{-1}[-\delta\gamma_{R,I},0,0]^{T}.
\\
\end{aligned}
\end{equation}
$\bold{A}$ is defined in~(\ref{A_def}) and other symbols are given in Table~\ref{var_list}.
\end{theorem}
\begin{proof}
The proof of Theorem~\ref{fin_dmt} is given in Appendix~\ref{proof_dmt}.
\end{proof}
\begin{remark}
Similar to~(\ref{de1}), the derivatives in~(\ref{de_z}) are obtained by the chain rule. It can be observed that $H(z)$ plays an important role in the finite-SNR DMT. The expression can be further simplified by an approximation, which is given in the following proposition.
\end{remark}

\begin{proposition} (\textit{An approximation for the finite-SNR DMT}) The finite-SNR DMT in~Theorem~\ref{fin_dmt} can be approximated by
\begin{equation}
d(m,\rho) \approx -\frac{z(m-k)^2H(z)H'(z)}{k^2}.
\end{equation}
\end{proposition}
\begin{remark}
This approximation is obtained by the approximation $Q(x)\le \frac{e^{-\frac{x^2}{2}}}{2},~x >0$~\cite{chiani2003new}. It can be observed that the finite-SNR DMT is highly related to the ratio between the mean and the standard deviation of the MI, i.e., $H(z)$. This indicates that the effect of the SNR on the finite-SNR DMT is represented through the term $-zH(z)H'(z)$.
\end{remark}

\section{How Large the IRS Needs to Be?}
\label{sec_siz}

\subsection{Impact of the Size of IRSs on EMI}
We first characterize the relation between the EMI and the size of the IRS. For simplicity, we assume that $M=N$ and focus on the parameter $\tau$,  i.e., the ratio between the number of antennas $N$ and the size of the IRS $L$. To study the performance gain obtained by increasing the size of the IRS, we focus on the case when $\tau<1$, which means that the size of the IRS is larger than the numbers of antennas at the transceiver. Results with $\tau>1$ can be handled similarly. Inspired by the concept of \textit{massive MIMO efficiency} proposed in~\cite{hoydis2013massive}, we define the concept called \textit{IRS efficiency} to characterize the efficiency of increasing the IRS size to achieve higher EMI. 
\begin{definition}
\label{def_eff}
For a given SNR $\rho$, the IRS efficiency $\eta \in (0,1)$ is defined as
\begin{equation}
\begin{aligned}
  \eta=\frac{\overline{I}(\rho)}{\overline{I}_{\infty}(\rho)},
\end{aligned}
 \end{equation}
 where $\overline{I}(\rho)$ denotes the EMI achieved by a given size of the IRS and $\overline{I}_{\infty}(\rho)$ represents the maximum throughput with infinitely large IRSs. Specifically, the IRS efficiency $\eta$ represents the percentage of the maximum throughput that is achieved by a given size of IRS.
  \end{definition}

Given $\eta$, we can determine the minimum size of the IRS needed to achieve at least the fraction $\eta$ of the asymptotic performance. To better illustrate the impact of the IRS size, we consider independent channels and ignore the optimization of phase shifts. According to Proposition~\ref{mean_pro}, for given $\rho$ and $M$, $\overline{I}(\rho)$ will be determined only by $\tau$. Next, we will investigate $\overline{I}_{\infty}(\rho)$.

\vspace{-0.3cm}

\subsection{Asymptotic Mean and Variance for the MI}
We first determine $\overline{I}_{\infty}(\rho)$ for a finite SNR $\rho$.
The roots of the equation in~(\ref{cubic_g}) can be given by Cardano's formula~\cite{zheng2016asymptotic} as
\begin{equation}
\begin{aligned}
\label{cubic_root_g}
g_{m}\!=\!\frac{e^{\frac{m\pi}{3}}\omega(\tau,\rho)+e^{-\frac{m\pi}{3}} {\omega}^{*}(\tau,\rho) -2 }{3}, m=1,2,3,
\end{aligned}
\end{equation}
where ${\omega}(\tau,\rho)$ and ${\omega}^{*}(\tau,\rho)$ are given by 
\begin{equation}
\label{cubic_w}
\begin{aligned}
&\omega(\tau,\rho)= (\sqrt{ (3\rho\tau-3\rho-1)^3+(1+\frac{9\rho}{2}+9\rho\tau)^2  }
\\
&+1+\frac{9\rho}{2}+9\rho\tau)^{\frac{1}{3}},
\\
&\omega(\tau,\rho)^{*}=\frac{(1-3\rho\tau+3\rho)}{\omega(\tau,\rho)}.
\end{aligned}
\end{equation}
$(\cdot)^{\frac{1}{3}}$ and $(\cdot)^{\frac{1}{2}}$ take any cubic root and square root, respectively.
It is straightforward to verify that the term under the square root operator is negative when $\tau=0$ and there are three real roots.
In fact, by Vieta's formulas, the only positive root is the one we desire. When the size of IRS goes to infinity, i.e., $\tau \rightarrow 0$, the asymptotic solution of the cubic equation is given by
\begin{equation}
\label{g_inf}
g_{\infty}=\frac{2(\mathrm{Re}\left\{\omega(0,\rho)\right\}-1 )}{3},
\end{equation}
where $\omega(0,\rho)$ is chosen as the one whose real and imaginary parts are both positive. In this case, the asymptotic EMI with $\tau \rightarrow 0$ is given by
\begin{equation}
\label{EMI_inf}
\begin{aligned}
 \overline{I}_{\infty}(\rho)
  &=2N\log(1+g_{\infty})+  \frac{N\rho}{(1+g_{\infty})^2} -\frac{2N g_{\infty}}{1+g_{\infty}}.
\end{aligned}
 \end{equation}
In fact, the asymptotic mean and variance can be obtained by an easier approach as they are equal to those of the MI for a single-hop MIMO system over independent Rayleigh channels, as shown in the following proposition.
\begin{proposition}
When $N=M$, the mean and variance of the MI of a single-hop and independent Rayleigh channel are given by:
\begin{equation}
\label{ray_mn}
\begin{aligned}
\overline{I}_{Rayleigh}(\rho)&=N\log( \delta+ 1+ \rho)-\frac{N\delta}{1+\delta},
\\
V_{Rayleigh}(\rho)&=\log[\frac{(1+\delta)^2}{2\delta+1}],
\end{aligned}
 \end{equation}
 where $\delta$ is the positive solution of $\delta^2+\delta-\rho=0$. It holds true that
 \begin{equation}
 \label{eq_ray_lim}
\overline{I}_{\infty}(\rho)= \overline{I}_{Rayleigh}(\rho), ~~V_{\infty}(\rho)= V_{Rayleigh}(\rho).
 \end{equation}
 \end{proposition}
\begin{proof} The expression in~(\ref{ray_mn}) can be obtained as a special case of Proposition 1 in~\cite{hachem2008new}. Then, we have $g_{\infty}=\delta$ since the cubic equation in~(\ref{cubic_g}) degenerates to a quadratic one, i.e., 
$g(1+g)=\rho$. Therefore, it follows from~(\ref{EMI_inf})
\begin{equation}
\begin{aligned}
 \overline{I}_{\infty}(\rho)
&=N \log[(1+g_{\infty})^2] +\frac{N g_{\infty}}{1+g_{\infty}} -\frac{2N g_{\infty}}{1+g_{\infty}}
\\
  &=N\log[1+g_{\infty}+g_{\infty}(1+g_{\infty})]- \frac{N g_{\infty}}{1+g_{\infty}}
  \\
  &=N\log(1+\rho+g_{\infty})- \frac{Ng_{\infty}}{1+g_{\infty}}
  \\
  &=\overline{I}_{Rayleigh}(\rho).
\end{aligned}
 \end{equation}
With similar manipulations, the asymptotic variance can be written as 
\begin{equation}
\begin{aligned}
 V_{\infty}(\rho)
  &=\log[\rho (1+g_{\infty})^2]-\log(\rho+2\rho g_{\infty})
  \\
  &=\log[\frac{(1+g_{\infty})^2}{2g_{\infty}+1}]= V_{Rayleigh}(\rho).
\end{aligned}
  \vspace{-0.3cm}
 \end{equation}
\end{proof}
\begin{remark}
(\ref{eq_ray_lim}) indicates that the maximum throughput and the asymptotic variance of the MI with the two-hop IRS channel are equal to those of a single-hop channel if we do not consider correlations and path loss. 
\end{remark}




Based on~(\ref{EMI_inf}), we can determine how large the IRS needs to be to achieve a certain percentage of the asymptotic performance. The results for the general correlated cases can be derived similarly by Theorem~\ref{mean_app}, but with an implicit expression. 
 The impact of the IRS size on the finite-SNR DMT is omitted here due to complex calculations and approximations. Nevertheless, it can be obtained similarly based on Theorem~\ref{fin_dmt}. Specifically, we have obtained the approximation for $g$, the EMI, and the variance. To derive the finite-SNR DMT over correlated channels, we need to calculate the derivatives of the EMI and the variance with respect to $\rho$ by Theorem~\ref{fin_dmt} and omit the higher order terms.






In the above derivation, we assume that the transmitter and receiver have the same number of antennas. The derivations for unequal cases can also be performed in similar ways. Furthermore, the correlated case can also be investigated numerically using Theorem~\ref{mean_app} and Theorem~\ref{clt_mi}.

\subsection{The Impact of the IRS Size in the High SNR Regime}
At high SNRs, the EMI and outage probability can be approximated as a more explicit function of $\tau$ compared with the cubic root form in~(\ref{cubic_root_g}), which is presented by the following theorem.
\begin{theorem} (High SNR case) When $\tau<1$ and $\rho \gg 1$, the only positive solution of the cubic equation in~(\ref{cubic_g}) can be approximated as
\label{high_app}
\begin{equation}
\begin{aligned}
g&=\sqrt{(1-\tau)\rho }+[\frac{1}{2(1-\tau)}-1]+o(1)
\\
&=a\rho^{\frac{1}{2}}+b+o(1).
\end{aligned}
\end{equation}
When $N$ is fixed and $ L \ll \sqrt{\rho}$, the mean and variance of the MI can be approximated by
\begin{equation}
\label{appro_h_snr}
\begin{aligned}
\overline{I}(\rho)&=N[\log(\rho)-2-(\frac{1}{\tau}-1)\log(1-\tau)+\frac{2}{a}\rho^{-\frac{1}{2}}] 
\\&+o(\rho^{-\frac{1}{2}}),  
\\
\text{and}
\\
V(\rho)&=\frac{1}{2}\log(\frac{\rho}{4a^2})+\frac{2a^2(1-b)-1}{2a^3}\rho^{-\frac{1}{2}}+o(\rho^{-\frac{1}{2}}),
\end{aligned}
\end{equation}
respectively. Furthermore, when $N$ is fixed and $L$ is comparable to $\sqrt{\rho}$ or far larger than $\sqrt{\rho}$, the mean and variance of the MI can be approximated by
\begin{equation}
\label{appro_h_snr_h_l}
\begin{aligned}
\overline{I}(\rho)&=N[\log(\frac{\rho}{e})-\frac{2N}{L}+2\rho^{-\frac{1}{2}}]+o(\rho^{-\frac{1}{2}}),
\\
V(\rho)&=\frac{1}{2}\log(\frac{\rho}{4})+\frac{N}{2L}+\rho^{-\frac{1}{2}}+o(\rho^{-\frac{1}{2}}).
\end{aligned}
\end{equation}
The outage probability can be further approximated by $P_{out}(R)\approx \Phi\left(\frac{R-\overline{I}(\rho)}{\sqrt{V(\rho)}}\right)$.
\end{theorem}
\begin{remark}
\label{h_limit_p}
\textbf{Asymptotic performance}:
At high SNRs, as $L$ goes to infinity, $\overline{I}(\rho)$ and $V(\rho)$ converge to an asymptotic value as shown in~(\ref{appro_h_snr_h_l}), which is dominated by the term $\log(\rho)$.
\end{remark}
\begin{remark}
\label{fast_con}
\textbf{Fast convergence}:
$\overline{I}(\rho)$ and $V(\rho)$ converge fast to the asymptotic values. When $L$ is large, the rate of convergence to the asymptotic performance is $O(\frac{1}{L})$. This indicates that the gain of EMI achieved by increasing the size of the IRS decreases when the IRS size is getting larger. 
\end{remark}
\begin{remark}
\label{snr_vs_irs}
\textbf{SNR vs. size of IRS}:
We can observe from~(\ref{appro_h_snr_h_l}) that the impact of the IRS size is very limited compared with that of the SNR as the mean and variance are dominated by the SNR.
\end{remark}
\begin{remark}
\textbf{Comparison with the single-hop Rayleigh channel}:
At high SNRs, the mean and variance of the MI for a single-hop i.i.d. Rayleigh channel is given in~\cite{loyka2010finite}:
\begin{equation}
\label{appro_h_snr_ray}
\begin{aligned}
\overline{I}(\rho)&=N[\log(\frac{\rho}{e})+2\rho^{-\frac{1}{2}}]+o(\rho^{-\frac{1}{2}}),
\\
V(\rho)&=\frac{1}{2}\log(\frac{\rho}{4})+\rho^{-\frac{1}{2}}+o(\rho^{-\frac{1}{2}}).
\end{aligned}
\end{equation}
Comparing the results in~(\ref{appro_h_snr_ray}) and~(\ref{appro_h_snr_h_l}), we can observe that the difference induced by the IRS is of order $O(\frac{1}{L})$, which indicates that these two results converge as the IRS size goes to infinity.
\end{remark}
The analysis in this section revealed the effect of increasing the size of the IRS, which provides guidance on the design of practical systems. In fact, as indicated by the square law~\cite{wu2019towards},~\cite{wu2019intelligent}, the impact of the IRS size is not only related to the physical size but also the passive beamforming gains introduced by designing phase shifts. In our analysis, the impact of the phase shifts is ignored as there are no closed-form optimal phase shifts for the correlated IRS-aided MIMO system.


\section{Numerical Results}
\label{simu}
In this section, simulations are performed to validate the theoretical results. Here we adopt the exponential correlation model. Specifically, the entries of $\bold{R}_{i}, \bold{T}_{i}, i=1,2$, can be represented by $[\bold{C}(\mu)]_{i,j}=\mu^{|i-j|}$. The unit of the preset transmission rate $R$ is bits/s/Hz. We assume the distance from the BS to the IRS and that from the IRS to the UE, denoted by $d_{BS-IRS}$ and $d_{IRS-UE}$, respectively, are both $10~\mathrm{m}$\footnote{The channel model utilized in this paper is proposed for the far-field region, i.e., the transimission distance is larger than the Rayleigh distance. As the MIMO/IRS size becomes larger, the Rayleigh distance will increase. The simulation setting is adopted from previous works and for illustration purpose only. Caution should be exercised when the antenna size is very large.}. The distance-dependent path loss is given by
\begin{equation}
L(d)=C_{0}(\frac{d}{D_{0}})^{-\alpha},
\end{equation}
where $C_{0}=-30~\mathrm{dB}$ is the path loss at the reference distance $D_{0}=1~\mathrm{m}$~\cite{wu2019intelligent}. $d$ denotes the distance of the link, and $\alpha$ denotes the path loss exponent. The path loss exponent of the link from the IRS to the UE is $\alpha_{IRS-UE}=3$ and that from the BS to the IRS is $\alpha_{BS-IRS}=2$~\cite{wu2019intelligent}. 
In the following figures, ``Ana'' and ``Sim'' will be used to represent the analytical results and Monte-Carlo simulations, respectively.

\subsection{Outage Probability Approximation}
First, we consider the case when there is no correlation at the transmitter side, for which both the proposed method and the method based on Mellin transform~\cite{shi2021outage} can be applied. The parameters are set to be $M=L=N=3$, $\bold{R}_{1}=\bold{R}_{2}=\bold{T}_{1}=\bold{C}(\mu)$ with $\mu=0.5$, and $\bold{T}_{2}=\bold{I}_{M}$. The power of the noise $\sigma^2=-114.7~ \mathrm{dBm}$. The phase shift matrix is set to be $\bold{\Psi}=\bold{I}_{L}$. The Monte-Carlo simulation values are computed over $10^{8}$ realizations, and we use the variance in~(\ref{low_var}) for the small $L$. It can be observed from Fig.~\ref{comp_app} that the performance of the proposed method is almost the same as the one based on Mellin transform but with a simpler expression. 

In Fig.~\ref{corr_low_rank}, we consider the correlated case with $\bold{T}_{2}=\bold{C}(0.5)$ and plot the outage probability with respect to the rate threshold when $L=3, 16, 32$, respectively. Markers represent the simulation results. It can be observed that the proposed method works well when there are correlations at both the transmitter and receiver sides. 

\subsection{Effectiveness of the Proposed Optimization Algorithm}
\begin{figure}[t!]
\centering\includegraphics[width=0.35\textwidth]{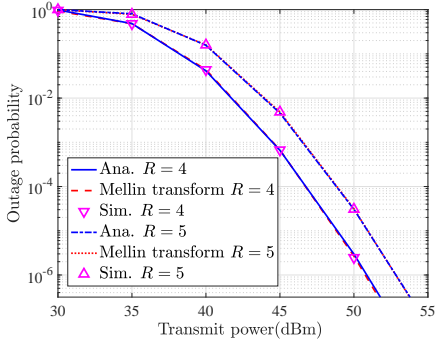}
\vspace{-0.3cm}
\caption{Performance comparison with the Mellin transform-based method in~\cite{shi2021outage}.}
\label{comp_app}
\vspace{-0.5cm}
\end{figure}
\begin{figure}[t!]
\centering\includegraphics[width=0.35\textwidth]{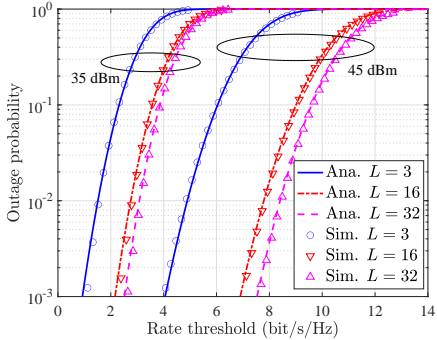}
\vspace{-0.3cm}
\caption{General correlated and rank-deficient cases.}
\label{corr_low_rank}
\vspace{-0.5cm}
\end{figure}
\begin{figure}[t!]
\centering\includegraphics[width=0.35\textwidth]{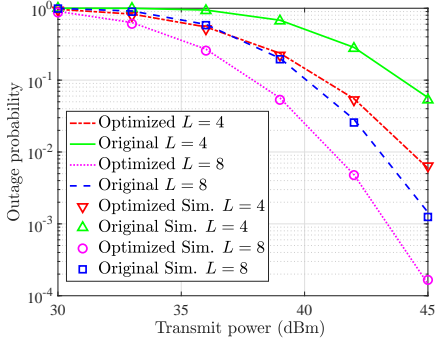}
\vspace{-0.3cm}
\caption{Effectiveness of the optimization algorithm.}
\label{Opt_VS_Ori}
\vspace{-0.5cm}
\end{figure}
\begin{figure}[t!]
\centering\includegraphics[width=0.35\textwidth]{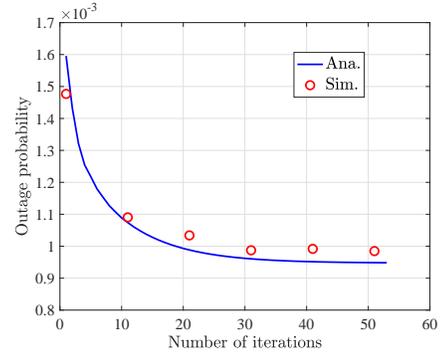}
\vspace{-0.3cm}
\caption{Convergence of Algorithm~\ref{gra_alg}.}
\label{iter_values}
\vspace{-0.3cm}
\end{figure}
Fig.~\ref{Opt_VS_Ori} verifies the effectiveness of Algorithm~\ref{gra_alg} in minimizing the outage probability where $M=N=4$ and $\mu=0.8$ is used for all correlation matrices. The power of the noise $\sigma^2=-116~ \mathrm{dBm}$. The initial phase shifts are set to be ${\psi}_{i}=e^{\frac{\jmath2\pi i}{L}}, i=1,2,...,L$. The initial step size is set as $\alpha_{0}=0.0005$ and the scale factor $\beta=0.5$. It can be observed that the outage probability is efficiently improved by Algorithm~\ref{gra_alg}. Fig.~\ref{iter_values} shows the outage probability in each outer iteration when $L=32$, from which  we can observe that the outage probability converges.

\subsection{Impact of Correlations at the Transceiver}
In Fig.~\ref{corr_tr}, we investigate the impact of the correlations at the transceiver. We set the correlation coefficients at the IRS to be $0.5$ with $L=16$, while the coefficients of the correlation matrices at the transceiver, i.e., $\bold{R}_1=\bold{T}_2$, are set to range from 0 to $0.9$ with $M=N=4$ and $\sigma^2=-116~ \mathrm{dBm}$. The outage probability is computed according to the phase shifts optimized by Algorithm~\ref{gra_alg}. It can be observed that the correlation at the transceiver has a negative effect on the outage performance. Similar observations have been reported in~\cite{shi2021outage},~\cite{zhang2021sum}.

\begin{figure}[t!]
\centering\includegraphics[width=0.35\textwidth]{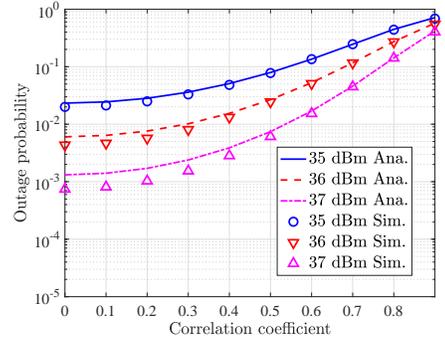}
\vspace{-0.3cm}
\caption{Impact of correlations at the transceiver.}
\label{corr_tr}
\vspace{-0.5cm}
\end{figure}

\subsection{Finite SNR DMT}
Fig.~\ref{dmt_r} validates the accuracy of the closed-form expression for the finite SNR DMT as shown in~(\ref{dmt_exp}). The parameters are set as $N=M=4$, $L=2$, $\mu=0.5$, $\bold{\Psi}=\bold{I}_{L}$, and $\sigma^2=-116~ \mathrm{dBm}$. It can be observed that the closed-form expression is very accurate. Furthermore, the finite-SNR DMT is more accurate at low SNRs compared with the asymptotic SNR results in~\cite{yang2011diversity}. 

\begin{figure}[t!]
\centering\includegraphics[width=0.35\textwidth]{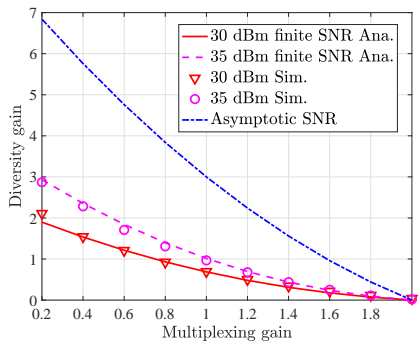}
\vspace{-0.3cm}
\caption{Diversity and multiplexing trade-off.}
\label{dmt_r}
\vspace{-0.5cm}
\end{figure}

\subsection{Impact of the Size of the IRS}
\textbf{Asymptotic Performance and Convergence}: 
In Fig.~\ref{lim_emi}, we show how the EMI changes when the IRS size increases. The parameters are set as $M=N=20$ with independent channels. The range of the size of IRS is from $4$ to $100$. It can be observed that the EMI increases and approaches the asymptotic performance as the IRS size becomes larger. However, the increasing rate of the EMI decreases, which indicates that the performance gain from the increment of the IRS size varnishes for larger IRS. This agrees with Remarks~\ref{h_limit_p} and~\ref{fast_con}. Similar comparisons are performed for the variance of the MI with $M=N=4$ and $R=4$. $\sigma^2=-116~ \mathrm{dBm}$. It can be observed from Fig.~\ref{var_irs} that the variance decreases and approaches the asymptotic value when the size of the IRS becomes larger and the approximations for the variance are accurate.


For the correlated case, we use the numerical results to show the bottleneck for the performance improvement induced by increasing the size of the IRS. In Fig.~\ref{cor_lim_throughput} and Fig.~\ref{cor_lim_var}, the EMI and the variance with the phase shifts optimized by Algorithm~\ref{gra_alg} are given. We have similar observations, i.e., both the EMI and variance will approach the limit as the size of the IRS increases.

\textbf{IRS Efficiency}: In Fig.~\ref{beta_irs}, the size of the IRS needed to achieve $\eta$ of the asymptotic performance is plotted when $M=N=20$. It is worth noting that the required size for $\eta=0.95$ is almost twice that for $\eta=0.9$, which means that the $5\%$ performance increment requires a huge increase in the size of the IRS. This agrees with Remark~\ref{fast_con}. In practice, we need to decide whether it is worth the cost of deploying larger IRSs to obtain the additional performance.

\textbf{High SNR Regime}: In Figs.~\ref{appro_outage}, the approximation performance of Theorem~\ref{high_app} for the outage probability in the high SNR regime is validated with $M=N=4$. The range of the IRS size is from $8$ to $256$ and the transmit power is set to be $50$dBm. It can be observed that the approximation in~(\ref{appro_h_snr}) is accurate in the high SNR regime. 


\textbf{SNR vs. IRS Size}: 
We can observe from Figs.~\ref{lim_emi} and~\ref{appro_outage} that, when the size of the IRS increases, the performance (EMI and outage probability) improves but the rate of change decreases, leading to the convergence to the asymptotic value. Furthermore, the EMI gain caused by the increment of SNR is far larger than that caused by the increment of the size of IRS, which agrees with Remark~\ref{snr_vs_irs}. 

\begin{figure}[t!]
\centering\includegraphics[width=0.35\textwidth]{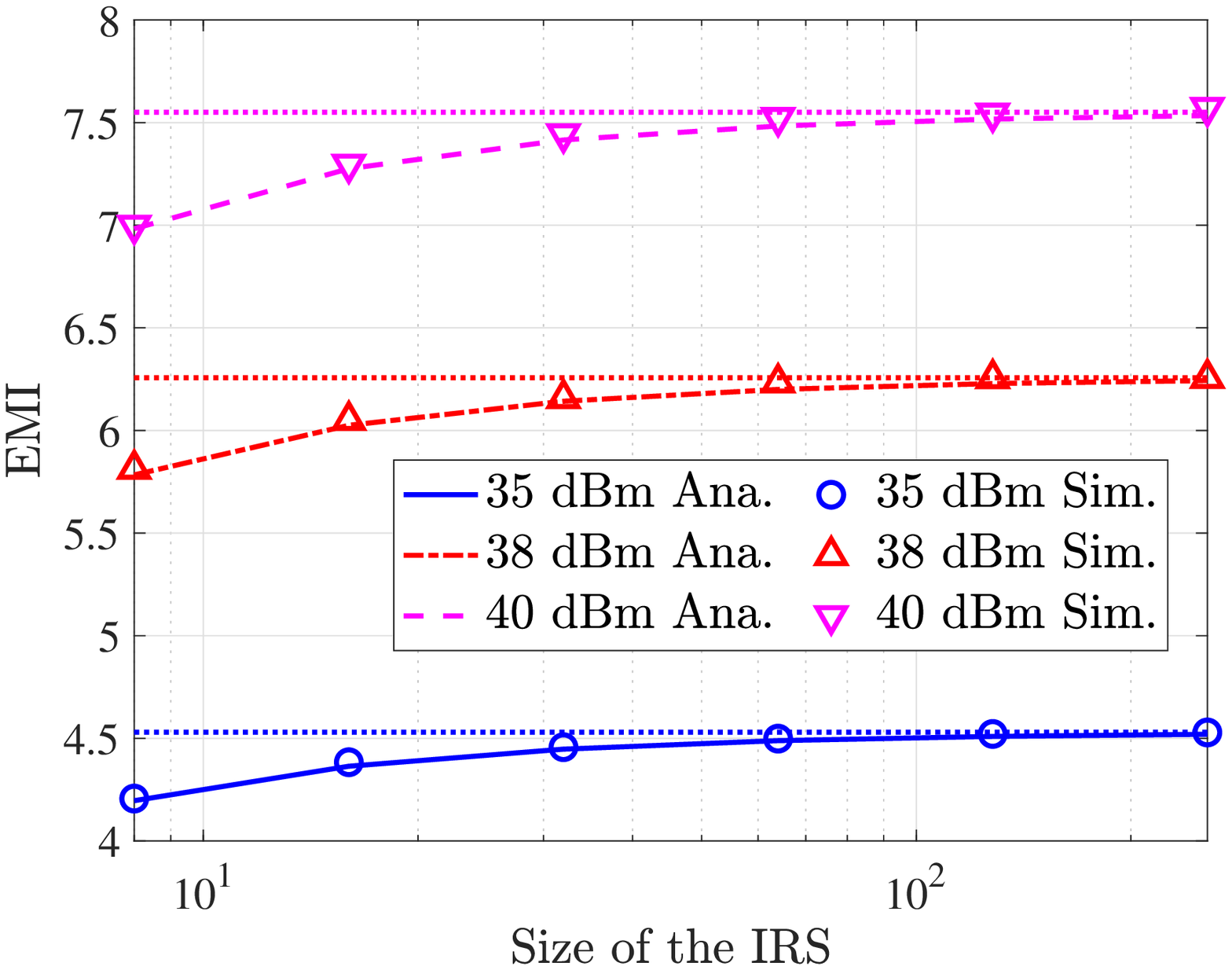}
\vspace{-0.3cm}
\caption{Impact of size on EMI.}
\label{lim_emi}
\vspace{-0.5cm}
\end{figure}
\begin{figure}[t!]
\centering\includegraphics[width=0.35\textwidth]{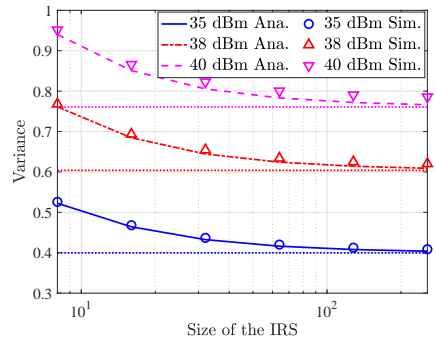}
\vspace{-0.3cm}
\caption{Asymptotic variance, variance v.s. size of IRS.}
\label{var_irs}
\vspace{-0.5cm}
\end{figure}
\begin{figure}[t!]
\centering\includegraphics[width=0.35\textwidth]{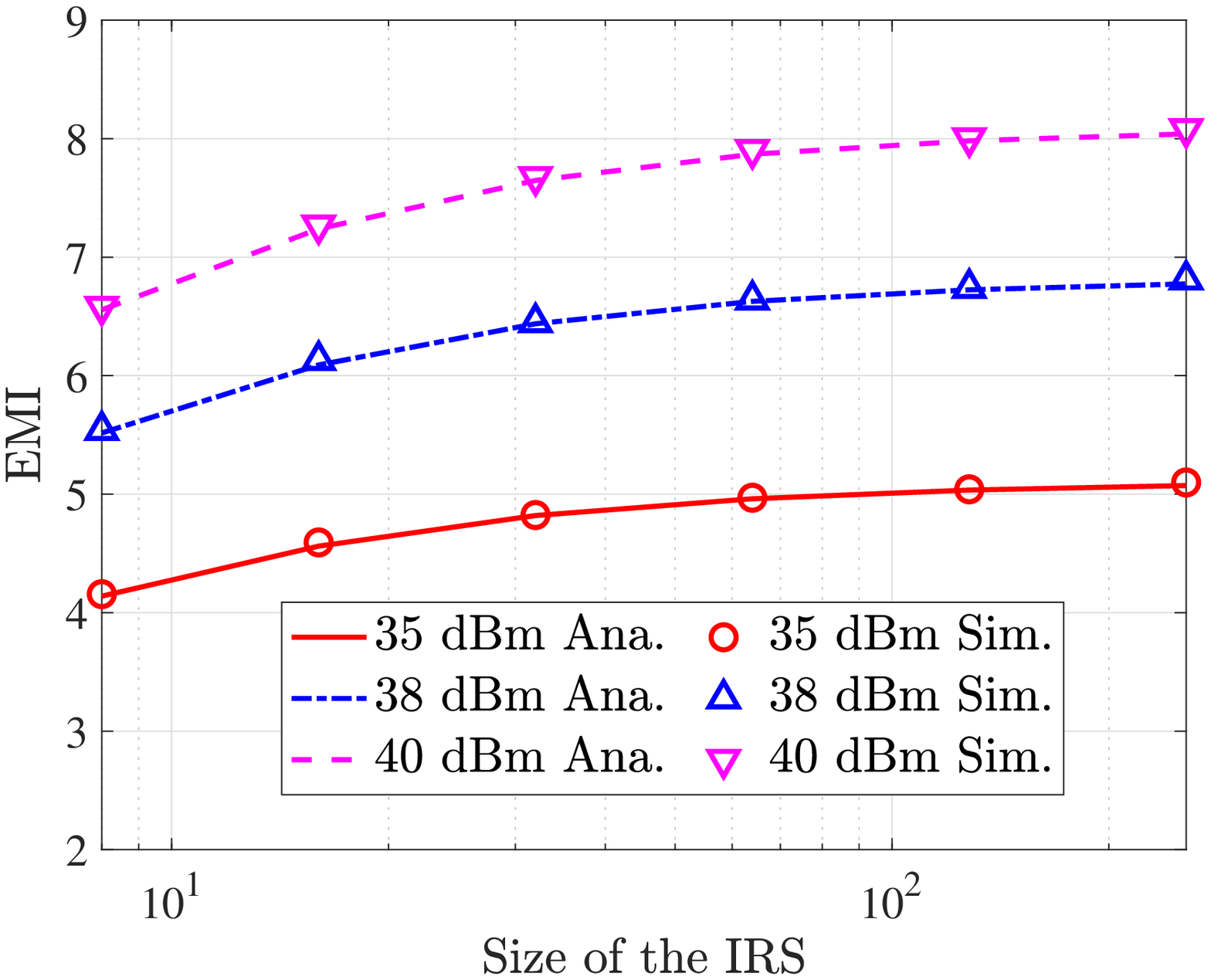}
\vspace{-0.3cm}
\caption{Impact of size on EMI.}
\label{cor_lim_throughput}
\vspace{-0.5cm}
\end{figure}
\begin{figure}[t!]
\centering\includegraphics[width=0.35\textwidth]{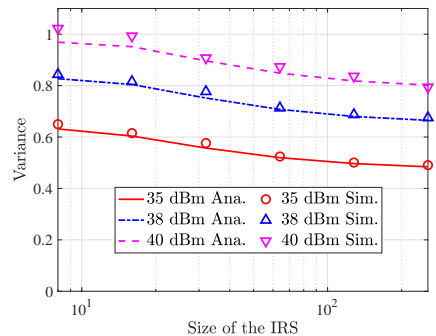}
\vspace{-0.3cm}
\caption{Variance v.s. size of IRS.}
\label{cor_lim_var}
\vspace{-0.5cm}
\end{figure}

\begin{figure}[t!]
\centering\includegraphics[width=0.35\textwidth]{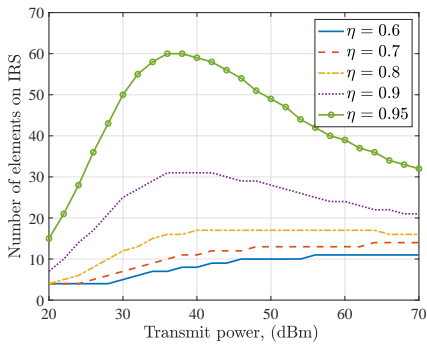}
\vspace{-0.3cm}
\caption{Size of IRS v.s. SNR given different $\eta$ s.}
\label{beta_irs}
\vspace{-0.4cm}
\end{figure}

\begin{figure}[t!]
\centering\includegraphics[width=0.35\textwidth]{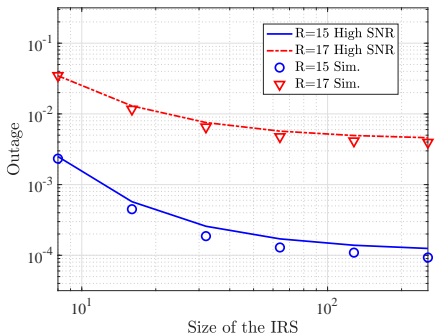}
\vspace{-0.3cm}
\caption{Outage approximation in the high SNR regime.}
\label{appro_outage}
\vspace{-0.5cm}
\end{figure}

\section{Conclusion}
\label{sec_con}
In this paper, we investigated the outage probability and finite-SNR DMT of IRS-aided MIMO systems assuming statistical CSI. By leveraging RMT, the CLT of the MI was derived and the outage probability was approximated in a closed-form. A gradient descent algorithm was proposed to minimize the outage probability by optimizing the phase shifts, which was shown to be efficient by numerical results. To better characterize the trade-off between the outage probability and the throughput in an operational SNR regime, a closed-form finite SNR DMT was presented. Finally, the impact of the IRS size on the EMI and outage probability was studied. Both theoretical and simulation results showed that larger IRSs provide better performance but the gain saturates quickly as the size of IRS increases. For example, to improve the throughput from 90\% to 95\% of its maximum value, a huge hardware cost (IRS size) is required. The results in this work not only provided accurate characterization for the MI and outage probability of IRS-aided systems, which is not available in the literature, but also revealed the impact of the IRS size on system performance, offering valuable guidance for practical system deployment. Furthermore, the result in this work can be extended to the case with the existence of the direct link. More efforts have to be involved since one more random matrix (the channel matrix of the direct link) needs to be handled in an iterative approach and more complex computations are required for calculating the asymptotic variance.


%

\appendices
\section{Proof of Theorem~\ref{clt_mi}}
\label{proof_clt}
\begin{proof}
Recall that $\bold{R},\bold{S},\bold{T}$ can be replaced by diagonal matrices consisting of their eigenvalue matrix. Therefore, $\bold{R},\bold{S},\bold{T}$ are assumed to be diagonal without loss of generality, i.e. $\bold{R}=\diag(r_1,r_2,...,r_{N})$, where $r_1,r_2,...,r_N$ are the eigenvalues of $\bold{R}$. Similarly, we have $\bold{S}=\diag(s_1,s_2,...,s_{L})$ and $\bold{T}=\diag(t_1,t_2,...,t_{M})$. As the channel matrix $\bold{H}$ in~(\ref{h_rst}) is the product of two random matrices, i.e., $\bold{X}$ and $\bold{Y}$, we use the martingale method to show the CLT~\cite{zheng2012central}. First, we rewrite the process $I(\rho)- \E I(\rho) $ as a summation of two processes:
\begin{equation}
[I(\rho)- \E( I(\rho)|\bold{Y})] +  [ \E( I(\rho)|\bold{Y}) -\E I(\rho)],
\end{equation}
whose variances are $V_{1}$ and $V_{2}$ respectively.
Then, we show that the two processes are asymptotically Gaussian so that $I(\rho)- \E I(\rho)$ will be Gaussian. In the following, we first provide the outline of the proof in steps and then give the details for each step.

\textit{Step 1:} Consider the first process $I(\rho)- \E( I(\rho)|\bold{Y})$ given $\bold{Y}=\left\{all~\bold{Y}\right\}$ and show that it converges to a Gaussian process. Furthermore, it is proved that the asymptotic mean and variance are independent of the condition $\bold{Y}$, which indicates that the asymptotic distribution of the first process is independent of that for the second process $\E( I(\rho)|\bold{Y}) -\E I(\rho)$ because its mean and variance are deterministic.

\textit{Step 2:} Show the asymptotic distribution of $\E( I(\rho)|\bold{Y}) -\E I(\rho)$ is Gaussian.

\textit{Step 3:} The sum of two independent Gaussian random processes will result in a Gaussian random process and we will compute the expression of the asymptotic variance in this step.

\textit{Step 1:}
Denote $\bold{W}=\bold{S}^{\frac{1}{2}}\bold{Y}\bold{T}^{\frac{1}{2}}$ and it follows from~(\ref{h_rst}) that the channel matrix can be given by $\bold{H}=\bold{R}^{\frac{1}{2}}\bold{X}\bold{W}$. $\|\bold{W}\bold{W}^{H} \|\le \|\bold{S} \| \|\bold{T} \|(1+\sqrt{\frac{L}{M}})    < \infty $ and $\lim\inf \frac{1}{L}\Tr\bold{Y}\bold{Y}^{H}>0$ hold true with probability one.
Therefore, the condition of Theorem 1 in~\cite{hachem2008new} is satisfied with probability one, from which we can conclude that
$I(\rho)- \E( I(\rho)|\bold{Y})$ converges to a zero-mean Gaussian process and the asymptotic variance $V_1$ is given by 
\begin{equation}
\label{first_vp1}
V_1=\Var(I(\rho)|\bold{Y}) \xrightarrow[M \rightarrow \infty]{\mathcal{P}} -\log(1-\gamma_{\bold{Y}}\widetilde{\gamma}_{\bold{Y}}),
\end{equation}
with
\begin{equation}
\begin{aligned}
\label{first_vv}
&\gamma_{\bold{Y}}=\frac{1}{L}\Tr\bold{R}\bold{C} \bold{R}\bold{C},\widetilde{\gamma}_{\bold{Y}}=\frac{1}{L}\Tr\bold{W}\bold{W}^{H}\widetilde{\bold{C}} \bold{W}\bold{W}^{H}\widetilde{\bold{C}},
\\
&\bold{C}=(z\bold{I}_{N}+\widetilde{f}\bold{R})^{-1},~\widetilde{\bold{C}}=(\bold{I}_{L}+f\bold{W}\bold{W}^{H})^{-1},\\
&f=\frac{1}{L}\Tr\bold{R}\bold{C},~\widetilde{f}=\frac{1}{L}\Tr\bold{W}\bold{W}^{H}\widetilde{\bold{C}}.
\end{aligned}
\end{equation}
The RHS of~(\ref{first_vp1}) can be further evaluated by the following lemma.
\begin{lemma}
\label{first_eva_v}
Assuming that \textbf{A.1}-\textbf{A.3} hold true and following the notations in~(\ref{first_vv}), there holds true that
\begin{equation}
\begin{aligned}
& -\log(1-\gamma_{\bold{Y}}\widetilde{\gamma}_{\bold{Y}})
 \xrightarrow[M \rightarrow \infty]{\mathcal{P}} 
 \\
 &-\log[1-\frac{M\gamma_{R}}{L\delta^2}(\frac{\gamma_{S}(\gamma_{T,I}^2-\frac{1}{M}\psi_{T} )}{\Delta_{Y}}
\!+\!\gamma_{T}g^2)]
\\
&=-\log(1-\gamma_{R}\Gamma_{L}),
\end{aligned}
\end{equation}
where $\gamma_{R}$ and $\Gamma_{L}$ are given in Table~\ref{var_list}.
\end{lemma}
\begin{proof}
The proof of Lemma~\ref{first_eva_v} is given in Appendix~\ref{other_app1}.
\end{proof}
By Lemma~\ref{first_eva_v}, we can obtain that
\begin{equation}
\label{first_v1}
V_1\xrightarrow[M \rightarrow \infty]{\mathcal{P}} -\log(1-\gamma_{R}\Gamma_{L}),
\end{equation}
where the RHS does not depend on $\bold{Y}$. According to the asymptotic regime and assumptions, we can conclude that $\delta, \Delta_{Y}, \gamma_{S}, \psi_{T}$ are all of order $\Theta(1)$. For example, it can be verified that $0<\frac{\Tr\bold{R}}{L(z+\|\bold{R} \|\|\bold{S} \|\|\bold{T} \| )}\le \delta \le \frac{N\|\bold{R} \|}{Lz} < \infty$ and $\Delta_{Y}, \gamma_{S}, \psi_{T}$ can be verified similarly. Therefore, the term $-\frac{\gamma_{S}\psi_{T} }{L\delta^2\Delta_{Y}}$ will vanish as $L$ goes to infinity and $\Gamma_{L}\xrightarrow[]{L \rightarrow \infty}\Gamma$.  For a better approximation of the variance, we will use $\Gamma_{L}$ for the cases with small $L$ and $\Gamma$ for the cases with large $L$, respectively.

Similar analysis can be performed on the asymptotic mean to show that the asymptotic mean also does not depend on $\bold{Y}$, which indicates that the asymptotic distribution of the first process is independent of that for the second process as the
asymptotic mean and variance are deterministic.

\textit{Step 2:} In this step, we will investigate $\E( I(\rho)|\bold{Y}) -\E I(\rho)$, which is a random variable with respect to $\bold{Y}$. We will show the asymptotic Gaussianity of $\E( I(\rho)|\bold{Y}) -\E I(\rho)$ using the martingale method (CLT for martingales) in~\cite{billingsley2008probability}. Let $\bold{Z}=\bold{R}^{\frac{1}{2}}\bold{X}\bold{S}^{\frac{1}{2}}$, then $\bold{H}=\bold{Z}\bold{Y}\bold{T}^{\frac{1}{2}}$. By the bookkeeping approach in~\cite{hachem2012clt} and~\cite{hachem2008clt}, we have
\begin{equation}
\begin{aligned}
&\E( I(\rho)|\bold{Y}) -\E_{\bold{Y}}\E ( I(\rho)|\bold{Y})
\\
&= \sum_{m=1}^{M} (\E_{m}-\E_{m-1})\log(1+\Lambda_{m}),
\end{aligned}
\end{equation}
where 
\begin{equation}
\begin{aligned}
\Lambda_{m}&=\frac{\bold{y}_{m}^{H}\bold{Z}^{H}\bold{Q}_{m}\bold{Z}\bold{y}_{m}- \frac{{t}_{m}}{M}\Tr \bold{Z}^{H}\bold{Z}\bold{Q}_{m} }{1+\frac{{t}_{m}}{M}\Tr\bold{Z} \bold{Z}^{H}\bold{Q}_{m}}
\\
&=\frac{e_{m}}{1+\frac{{t}_{m}}{M}\Tr \bold{Z}\bold{Z}^{H}\bold{Q}_{m}}.
\end{aligned}
\end{equation}
Here $\E_{m}(\cdot)=\E(\cdot|\bold{y}_{m},\bold{y}_{m+1},...,\bold{y}_{M} ) $, $\E_{M+1}(\cdot)=\E(\cdot)$, and $\bold{y}_{i}$ denotes the $i$-th column of $\bold{Y}$.

We skip the verification of the Lyapunov's condition and turn to the computation of the asymptotic variance for the second process, $V_2$. Following a similar method as in~\cite{hachem2012clt},~\cite{hachem2008clt}, we have 
\begin{equation}
\begin{aligned}
\label{v2_1}
V_{2}&\xlongrightarrow[M \longrightarrow \infty]{\mathcal{P} } \sum_{m=1}^{M}\E_{m-1}(\E_{m} \Lambda_{m} )^2
\\
&= \sum_{m=1}^{M}[\bold{Q}_{T}]_{m,m}^2\E_{m-1}(\E_{m}(e_{j})^2)
\\
&\overset{(a)}{=} \sum_{m=1}^{M}[\bold{Q}_{T}]_{m,m}^2\E_{m-1}( \frac{t_{m}^2}{M^2}\E_{m}\Tr[ \E_{m}(\bold{Z}^{H}\bold{Q} \bold{Z})
\\
& \times
\bold{Z}^{H}\bold{Q} \bold{Z} ])+o(1),
\end{aligned}
\end{equation}
where step $(a)$ follows from the rank-one lemma~\cite[Lemma 3.1]{hachem2012clt}). 
\begin{lemma} 
\label{zqz_comp}
Let $\bold{B}$ and $\bold{D}$ be deterministic matrices with bounded norm. For $\bold{Q}=(\bold{Z}\bold{D}\bold{Z}^{H}+z\bold{I}_{N})^{-1}$, where $\bold{Z}=\bold{R}^{\frac{1}{2}}\bold{X}\bold{S}^{\frac{1}{2}}$ is defined in Appendix~\ref{proof_clt}, it holds true that
\begin{equation}
\begin{aligned}
 &\frac{1}{M}\E\Tr\bold{Z}^{H}\bold{Q}\bold{Z}\bold{B}= \frac{1}{M}\E\Tr\bold{S}\bold{B}
 \\
&\times
 ((\frac{1}{L}\Tr\E\bold{R}\bold{Q})^{-1}\bold{I}_{L}+ \bold{S} \bold{D} )^{-1} 
+O(\frac{\| \bold{D} \|^2}{M}).
\end{aligned}
\end{equation}
\end{lemma} 
By using Lemma~\ref{zqz_comp} in Appendix~\ref{other_app2} twice 
and similar computations in~\cite{hachem2012clt}, we have 
\begin{equation}
\label{usage_lm}
\begin{aligned}
&\E_{m}\frac{1}{M}\Tr[ \E_{m}(\bold{Z}^{{H}}\bold{Q} \bold{Z})\bold{Z}^{{H}}\bold{Q} \bold{Z} ]
   \xlongrightarrow[M \longrightarrow \infty]{\mathcal{P} ~(a)} 
   \\
   &\frac{1}{M}\E_{m}\Tr[\E_{m}(((\frac{1}{L}\Tr\E_{\bold{X}}\bold{R}\bold{Q})^{-1}\bold{I}_{L}+ \bold{S} \bold{Y}\bold{T}\bold{Y}^{H} )^{-1})
   \\
    & \times \bold{S}\bold{Z}^{H}\bold{Q} \bold{Z}  ] \xlongrightarrow[M \longrightarrow \infty]{\mathcal{P}~(b) }   
   \\
   &\frac{1}{M}\E_{m}\Tr[\E_{m}(((\frac{1}{L}\Tr\E_{\bold{X}}\bold{R}\bold{Q})^{-1}\bold{I}_{L}+ \bold{S} \bold{Y}\bold{T}\bold{Y}^{H} )^{-1})
   \\
    & \times \bold{S}((\frac{1}{L}\Tr\E_{\bold{X}}\bold{R}\bold{Q})^{-1}\bold{I}_{L}+ \bold{S} \bold{Y}\bold{T}\bold{Y}^{H} )^{-1} \bold{S}],
\end{aligned}
\end{equation}
where $\E_{\bold{X}}$ represents the expectation with respect to $\bold{X}$. Step $(a)$ follows by taking $\bold{D}={\bold{Y}}\bold{T}\bold{Y}^{H}$ and $\bold{B}=\bold{Z}^{H}\bold{Q} \bold{Z}$ in Lemma~\ref{zqz_comp} as shown in Appendix~\ref{other_app2} and the convergence follows from
$$\E | \mathrm{LHS}-\mathrm{RHS}|\le \frac{K}{M^2} \E ( \| \bold{Y}\bold{T}\bold{Y}^{H}  \|^2  )=O(M^{-2})$$ according to Lemma 2 in~\cite{rubio2012clt} and Markov's inequality, where $\mathrm{LHS}$ denotes the left-hand side of step $(a)$. Similarly, step $(b)$ follows by taking $\bold{B}=\E_{m}(\frac{1}{L}\Tr\E_{\bold{X}}\bold{R}\bold{Q})^{-1}\bold{I}_{L}+ \bold{S} \bold{Y}\bold{T}\bold{Y}^{H} )^{-1}\bold{S}$. Therefore, we have
\begin{equation}
\begin{aligned}
   &\frac{1}{M}\E_{m}\Tr\E_{m}(((\frac{1}{L}\Tr\E_{\bold{X}}\bold{R}\bold{Q})^{-1}\bold{I}_{L}+ \bold{S} \bold{Y}\bold{T}\bold{Y}^{H} )^{-1})
   \\
   &
    \bold{S}((\frac{1}{L}\Tr\E_{\bold{X}}\bold{R}\bold{Q})^{-1}\bold{I}_{L}+ \bold{S} \bold{Y}\bold{T}\bold{Y}^{H} )^{-1} \bold{S}
    \xlongrightarrow[M \longrightarrow \infty]{\mathcal{P}~(a) } 
    \\
& \frac{1}{M}\E_{m}\Tr\E_{m}((\delta^{-1}\bold{I}_{L}+ \bold{S} \bold{Y}\bold{T}\bold{Y}^{H} )^{-1} )\bold{S}
\\
&\times (\delta^{-1}\bold{I}_{L}+ \bold{S} \bold{Y}\bold{T}\bold{Y}^{H} )^{-1} \bold{S}
\\
&    \xlongrightarrow[M \longrightarrow \infty]{\mathcal{P} ~(b)} 
 \frac{1}{M}\E\Tr\bold{S}\E_{m}[\widetilde{\bold{C}}(\delta^{-1})]\bold{S}\widetilde{\bold{C}}(\delta^{-1})\overset{\Delta}{=}\bold{\Omega}_{m},
\end{aligned}
\end{equation}
Steps $(a),~(b)$ follow from $\Var(\Tr\bold{R}\bold{Q})=O(1)$ and $\Var(\Tr\bold{S}\widetilde{\bold{C}}(\delta^{-1}))=O(1)$, respectively, which can be obtained by the variance control using Poincar{\`e}-Nash inequality~\cite{hachem2008new}~\cite{kammoun2019asymptotic}. 
According to~(\ref{v2_1}) and~\cite[Section 5]{hachem2012clt}, it holds true that
\begin{equation}
\begin{aligned}
\bold{\Omega}_{m}=\frac{\gamma_{S}}{1-\gamma_{S}\gamma_{T}^{(m)}}+o(1),
\end{aligned}
\end{equation}
where $\gamma_{T}^{(m)}=\frac{1}{M}\sum\limits_{l=1}^{m}[\bold{Q}_{T}]_{l,l}^{2}t_{l}^2$. Also, we have $\gamma_{T}^{(m)}-\gamma_{T}^{(m-1)}=\frac{1}{M}[\bold{Q}_{T}]_{m,m}^{2}t_{m}^2$. Therefore, 
\begin{equation}
\begin{aligned}
V_{2} &\xlongrightarrow[M \longrightarrow \infty]{\mathcal{P} }   \sum_{m=1}^{M}\E_{m-1}(\E_{m} \Lambda_{m} )^2
\xlongrightarrow[M \longrightarrow \infty]{\mathcal{P} }  
\\
&\sum_{m=1}^{M}\frac{\gamma_{S}(\gamma_{T}^{(m)}-\gamma_{T}^{(m-1)}) }{1-\gamma_{S}\gamma_{T}^{(m)}}
\xlongrightarrow[]{M \longrightarrow \infty}
\\
 &-\log(1-\gamma_{T}\gamma_{S})=-\log\Delta_{Y}.
\end{aligned}
\end{equation}

\textit{Step 3:} Since $I(\rho)- \E( I(\rho)|\bold{Y})$ and $\E( I(\rho)|\bold{Y}) -\E I(\rho)$ are asymptotically independent Gaussian random processes, the asymptotic distribution of their sum is Gaussian and the variance is given by
\begin{equation}
\begin{aligned}
V=V_1+V_2.
\end{aligned}
\end{equation}
This completes the proof.
\vspace{-0.2cm}
\end{proof}

\section{Proof of Proposition~\ref{clt_ray}}
\label{proof_ray}
\begin{proof}
As this proposition is the special case of Theorem~\ref{clt_mi}, in which the Gaussianity has been proved, we only need to compute the variance. By~(\ref{basic_eq}), we have
\begin{equation}
\begin{aligned}
\overline{g}=\frac{1}{{g}+1},~
g=\frac{1}{\tau(\frac{1}{\delta}+\overline{g})},~
z\delta=\tau\overline{g},
\end{aligned}
\end{equation}
and 
\begin{equation}
\begin{aligned}
\gamma_{R}= \frac{\delta^2}{\tau}, ~\gamma_{S}=\tau g^2, ~\gamma_{T}=\overline{g}^2.
\end{aligned}
\end{equation}
Then, the variance can be written as
\begin{equation}
\begin{aligned}
&V(\rho)=-\log[(1+g)^2 -\tau g^2 -\tau g^2 \overline{g}^2-g^2(1-\tau g^2\overline{g}^2) ]
\\
&-\log(\overline{g}^2)=-\log(T_1)-\log(T_2),
\end{aligned}
\end{equation}
where
\begin{equation}
\begin{aligned}
T_1&=1+2g-\tau g^2 + \tau g^2(g-1)\overline{g}
\\
&=1+2g -2 \tau g^2\overline{g}
\\
&\overset{(a)}{=}1+2zg^3+2zg^2,
\end{aligned}
\end{equation}
and $(a)$ follows from $g^3+2g^2+(1+\tau/z-1/z)g-1/z=0$. The conclusion follows by noticing that $T_{2}=(1+g)^{-2}$.
\vspace{-0.4cm}
\end{proof}

\section{Proof of Theorem~\ref{fin_dmt}}
\label{proof_dmt}
\begin{proof} 
By the definition of multiplexing gain in~\cite{loyka2010finite}, we can obtain the rate $R$ as, 
\begin{equation}
R=\frac{m \overline{I}(\rho)}{k},
\end{equation} 
For the ease of manipulations, we use $\overline{I}(z)$ and $V(z)$, which are obtained by letting $z=\frac{1}{\rho}$ in $\overline{I}(\rho)$ and $V(\rho)$, respectively. Then, by the definition of finite-SNR DMT in~\cite{loyka2010finite}, we have
\begin{equation}
\begin{aligned}
d(m,\rho)&=-\frac{\partial  \log(P_{out}(\frac{(m-k)\overline{I}(z)}{k\sqrt{V(z)}})) }{\partial \log(\rho)}  
\\
&=\frac{\partial  \log(P_{out}(\frac{(m-k)H(z)}{k})) }{\partial z} z
\\
&=\frac{z(m-k)}{k\sqrt{2\pi}}\frac{\exp(-\frac{(m-k)^2H^2(z)}{2k^2})H'(z)}{\Phi(\frac{(m-k)H(z)}{k} )}.
\end{aligned}
\end{equation}
$H'(z)$ can be obtained by the chain rule.
\end{proof}

\section{Proof of Preliminary Results}
\label{other_app}
\subsection{Proof of Lemma~\ref{first_eva_v}}
\label{other_app1}
\begin{proof}
Next, we will show that, when $M \rightarrow \infty$, 
$V_1$ converges to a constant in probability. By the following equation derived from the Sherman-Morrison-Woodbury formula~\cite{horn2012matrix},
\begin{equation}
\begin{aligned}
\bold{q}^{H}\left(\bold{B}+\mu\bold{q}\bold{q}^{H}  \right)^{-1}=\frac{1}{1+\mu\bold{q}^{H}\bold{B}^{-1}\bold{q}}\bold{q}^{H}\bold{B}^{-1}
\end{aligned}
\end{equation}
we have~(\ref{long_eq}) at the top of the next page.
\begin{figure*}[!htbp]
\centering
\begin{equation}
\label{long_eq}
\begin{aligned}
&\widetilde{\gamma}_{\bold{Y}}=\frac{1}{L}\Tr\bold{W}\bold{W}^{H}\widetilde{\bold{C}}\bold{W}\bold{W}^{H}\widetilde{\bold{C}}
=\frac{1}{L}\sum_{i=1}^{L}\frac{ t_{i}\bold{y}^{H}_{i} \bold{S}^{\frac{1}{2}} \widetilde{\bold{C}}_{i}\bold{W}\bold{W}^{H}\widetilde{\bold{C}} \bold{S}^{\frac{1}{2}}\bold{y}_{i}}{1+ t_{i}f\bold{y}^{H}_{i} \bold{S}^{\frac{1}{2}} \widetilde{\bold{C}}_{i} \bold{S}^{\frac{1}{2}}\bold{y}_{i}  } 
=\frac{1}{L}\sum_{i=1}^{L}\frac{ t_{i} \bold{y}^{H}_{i} \bold{S}^{\frac{1}{2}} \widetilde{\bold{C}}_{i}\bold{W}\bold{W}^{H}\widetilde{\bold{C}} \bold{S}^{\frac{1}{2}}\bold{y}_{i}}{1+ \frac{t_{i}f}{M}\Tr\bold{S}\widetilde{\bold{C}} } + o(1)
\\
&=\frac{1}{L}\sum_{i=1}^{L}\frac{ t_{i} \bold{y}^{H}_{i} \bold{S}^{\frac{1}{2}} \widetilde{\bold{C}}_{i}\bold{W}_{i}\bold{W}_{i}^{H}\widetilde{\bold{C}}_{i} \bold{S}^{\frac{1}{2}}\bold{y}_{i} }{(1+ \frac{t_{i}f}{M}\Tr\bold{S}\widetilde{\bold{C}})^2 }
+ \frac{ t_{i}^2\bold{y}^{H}_{i} \bold{S}^{\frac{1}{2}} \widetilde{\bold{C}}_{i}\bold{S}^{\frac{1}{2}} \bold{y}_{i}\bold{y}_{i}^{H}\bold{S}^{\frac{1}{2}} \widetilde{\bold{C}}_{i} \bold{S}^{\frac{1}{2}}\bold{y}_{i}}{(1+ \frac{t_{i}f}{M}\Tr\bold{S}\widetilde{\bold{C}})^{2} }   + o(1)
=\frac{1}{L}\sum_{i=1}^{L}\frac{ \frac{t_{i}}{M}\Tr\bold{S} \widetilde{\bold{C}}_{i}\bold{W}_{i}\bold{W}_{i}^{H}\widetilde{\bold{C}}_{i}}{(1+ \frac{t_{i}f}{M}\Tr\bold{S}\widetilde{\bold{C}})^2 }
\\
&+ \frac{ t_{i}^2(\frac{1}{M}\Tr\bold{S}\widetilde{\bold{C}})^2 }{(1+ \frac{t_{i}f}{M}\Tr\bold{S}\widetilde{\bold{C}})^{2} }  + o(1)
=Y_{1}+Y_{2}+o(1).
\end{aligned}
\end{equation}
\hrulefill
\vspace{-0.5cm}
\end{figure*}

Next, we evaluate the two terms $Y_{1}$ and $Y_{2}$ in~(\ref{long_eq}). The numerator of $Y_1$ can be given by 
\begin{equation}
\begin{aligned}
&\frac{1}{M}\Tr\bold{S} \widetilde{\bold{C}}_{i}\bold{W}_{i}\bold{W}_{i}^{H}\widetilde{\bold{C}}_{i}
=\sum_{j \ne i}^{M}\frac{t_{j}}{M} \bold{y}^{H}_{j}\bold{S}^{\frac{1}{2}}\widetilde{\bold{C}}_{i} \bold{S} \widetilde{\bold{C}}_{i}\bold{S}^{\frac{1}{2}}\bold{y}_{j}
\\
&\overset{a.s.}{\longrightarrow}\sum_{j \ne i}^{M}\frac{t_{j}\bold{y}^{H}_{j}\bold{S}^{\frac{1}{2}}\widetilde{\bold{C}}_{ij} \bold{S} \widetilde{\bold{C}}_{ij}\bold{S}^{\frac{1}{2}}\bold{y}_{j}}{M(1+\frac{t_jf}{M}\Tr\bold{S}\widetilde{\bold{C}}_{ij}  )^2} 
\\
&\overset{a.s.}{\longrightarrow} \sum_{j \ne i}^{M}\frac{\frac{t_{j}}{M}\Tr \bold{S}\widetilde{\bold{C}}_{ij} \bold{S} \widetilde{\bold{C}}_{ij}}{M(1+\frac{t_jf}{M}\Tr\bold{S}\widetilde{\bold{C}}_{ij}  )^2},
\end{aligned}
\end{equation}
where $\bold{W}_{i}$ is obtained by removing the $i$-th column from $\bold{W}$. These two almost convergences come from the convergence of the quadratic form~\cite[Lemma~3]{mueller2016linear}. Following the steps in Appendix E of~\cite{hoydis2011iterative}, we can prove that for any matrix $\bold{U}$ with a bounded norm, it holds true that
\begin{equation}
\frac{1}{M}\Tr\bold{U}\widetilde{\bold{C}} \xrightarrow[M \rightarrow \infty]{a.s.} \frac{1}{M}\Tr\bold{U}{\bold{Q}}_{S}.
\end{equation}
Therefore, by similar arguments as~\cite{hoydis2011asymptotic}, we have
\begin{equation}
\label{intermediate_f}
\begin{aligned}
& f=\frac{1}{M}\Tr\bold{R}(z \bold{I}_{N}+\widetilde{f}\bold{R} )^{-1} \xrightarrow[M \rightarrow \infty]{a.s.}
\\
&\frac{1}{M}\Tr\bold{R}(z\bold{I}_{N}+\frac{g\overline{g}}{\delta}\bold{R})^{-1}=\delta,
\\
&\frac{f}{M}\Tr\bold{S}\widetilde{\bold{C}} \xrightarrow[M \rightarrow \infty]{a.s.} \frac{1}{M}\Tr\bold{S}{\bold{Q}}_{S}=g,
\\
&\frac{1}{M}\Tr \bold{T}(\bold{W}^{H}\bold{W}+ \bold{I}_{M})^{-1} \xrightarrow[M \rightarrow \infty]{a.s.} \frac{1}{M}\Tr\bold{T}{\bold{Q}}_{T}=\overline{g}.
\end{aligned}
\end{equation}
From~\cite[Theorem 1]{couillet2011deterministic},~\cite[Theorem 1]{wen2012deterministic}, or~\cite[Lemma 3]{kammoun2019asymptotic}, we can show that if $(h, \widetilde{h})$ is the solution of 
\begin{equation}
\begin{aligned}
h&=\frac{1}{M}\Tr\delta\bold{S}\left(\alpha \bold{I}_{L}+l\bold{D}+ \widetilde{h}\delta\bold{S} \right)^{-1},
\\
\widetilde{h}&=\frac{1}{M}\Tr\bold{T}\left(\bold{I}_{M}+h\bold{T} \right)^{-1},
\end{aligned}
\end{equation}
there holds true that
\begin{equation}
\label{md_1}
\begin{aligned}
&\frac{f}{M}\Tr\bold{S}\left(\alpha \bold{I}_{L}+l\bold{D}+f\bold{W}\bold{W}^{H} \right)^{-1} 
\\
&\xrightarrow[M \rightarrow \infty]{a.s.}
\frac{\delta}{M}\Tr\bold{S}\left(\alpha \bold{I}_{L}+l\bold{D}+ \widetilde{h}\delta\bold{S} \right)^{-1}.
\end{aligned}
\end{equation}
Then by letting $\alpha=1$ and $l=0$, we have $h=g$ and $\widetilde{h}=\overline{g}$. 
Taking the derivative of $h$ and $\widetilde{h}$ with respect to $l$, we have
\begin{equation}
\label{tde1}
\begin{aligned}
&\frac{\mathrm{d} \frac{f}{M}\Tr \bold{S}(\alpha \bold{I}_{L}+l\bold{D}+f\bold{W}\bold{W}^{H})^{-1}}{\mathrm{d}l} |_{l=0}
\xrightarrow[M \rightarrow \infty]{a.s.} \frac{\mathrm{d}{h}}{\mathrm{d}l}|_{l=0} 
\\
&=\frac{-1}{M}\Tr\bold{S}\left(\delta^{-1} \bold{I}_{L}+ \widetilde{h}\bold{S} \right)^{-1}
\bold{S}\left(\delta^{-1}\bold{I}_{L}+ \widetilde{h}\bold{S} \right)^{-1} \frac{\mathrm{d}\widetilde{h}}{\mathrm{d}l}
\\
&+\frac{-1}{M \delta}\Tr\bold{S}\left(\delta^{-1}\bold{I}_{L}+ \widetilde{h}\bold{S} \right)^{-1}
\bold{D}\left(\delta^{-1}\bold{I}_{L}+ \widetilde{h}\bold{S} \right)^{-1},
\\
&\frac{\mathrm{d}\widetilde{h}}{\mathrm{d}l} |_{l=0}
= \frac{-1}{M}\Tr\bold{T}\left(\bold{I}_{M}+h\bold{T} \right)^{-1}\bold{T}\left(\bold{I}_{M}+h\bold{T} \right)^{-1}  
 \frac{\mathrm{d} {h}}{\mathrm{d}l}.
\end{aligned}
\end{equation}

Therefore, we have the following system of equations
\begin{equation}
\label{sys_fir}
\begin{bmatrix} 
1 & \gamma_{S}  \\  
 \gamma_{T}  & 1
\end{bmatrix}
\begin{bmatrix} 
 \frac{\mathrm{d}{h}}{\mathrm{d}l}
 \\
\frac{\mathrm{d}\widetilde{h}}{\mathrm{d}l}
\end{bmatrix}
=
\begin{bmatrix} 
-\frac{\gamma_{S}(\bold{D})}{\delta}
 \\
0
\end{bmatrix}+o(1).
\end{equation}
According to~(\ref{intermediate_f}) and letting $\bold{D}=\bold{S}$, we have
\begin{equation}
\frac{1}{M}\Tr \bold{S}\widetilde{\bold{C}}\bold{S}\widetilde{\bold{C}}=-\frac{1}{f} \frac{\mathrm{d}{h}}{\mathrm{d}l}|_{l=0} \xrightarrow[M \rightarrow \infty]{a.s.}\frac{1}{\delta^2}\frac{\gamma_{S}}{1-\gamma_{S}\gamma_{T}}.
\end{equation}
Hence, we have
\begin{equation}
\label{omit_term}
\begin{aligned}
Y_1  \xrightarrow[M \rightarrow \infty]{a.s.}\frac{M\gamma_{S}(\gamma_{T,I}^2-\frac{1}{M}\psi_{T} )}{L\delta^2\Delta_{Y}},
~Y_2  \xrightarrow[M \rightarrow \infty]{a.s.}\frac{M\gamma_{T}g^2}{L\delta^2},
\end{aligned}
\end{equation}
where $\psi_{T}$ is given in Table~\ref{var_list}. By applying the continuous mapping theorem~\cite{billingsley2008probability} to logarithm function, the asymptotic variance is given by
\begin{equation}
\begin{aligned}
&V_1  \xrightarrow[M \rightarrow \infty]{\mathcal{P}} \!-\! \log[1 \!-\! \frac{M\gamma_{R}}{L\delta^2}(\frac{\gamma_{S}(\gamma_{T,I}^2-\frac{1}{M}\psi_{T} )}{\Delta_{Y}}
\!+\!\gamma_{T}g^2)]
\\
&=-\log(1-\gamma_{R}\Gamma_{L}).
\end{aligned}
\end{equation}
\end{proof}

\subsection{Proof of Lemma~\ref{first_eva_v}}
\label{other_app2}
\begin{proof} Here we omit the condition on $\bold{Y}$. With the integration-by-parts formula~\cite{hachem2008new}, we have
\begin{equation}
\begin{aligned}
&\E [\bold{Z}^{H}\bold{Q}\bold{Z}\bold{B}]_{i,j}
=\E \bold{z}_{i}^{H}\bold{Q}\bold{Z}\bold{b}_{j}
= \E s_{i}^{\frac{1}{2}} \bold{x}_{i}^{H}\bold{R}^{\frac{H}{2}} \bold{Q}\bold{R}^{\frac{1}{2}}\bold{X}\bold{S}^{\frac{1}{2}} \bold{b}_{j}
\\
&=\frac{1}{L}  \E s^{\frac{1}{2}}_{i} \sum_{m,n} {x}_{m,i}^{*}[\bold{R}^{\frac{H}{2}} \bold{Q}\bold{R}^{\frac{1}{2}}\bold{X}\bold{S}^{\frac{1}{2}}]_{m,n} {b}_{n,j}
\\
&=\frac{1}{L}\E   \sum_{m} -[ \bold{R}^{\frac{H}{2}} \bold{Q}\bold{R}^{\frac{1}{2}}]_{m,m} [ \bold{S}\bold{D}\bold{S}^{\frac{1}{2}}\bold{X}^{H}   \bold{R}^{\frac{H}{2}}\bold{Q}\bold{R}^{\frac{1}{2}}\bold{X}\bold{S}^{\frac{1}{2}}\bold{B}]_{i,j}
\\
&+\frac{1}{L} \E  \sum_{m} [ \bold{R}^{\frac{H}{2}} \bold{Q}\bold{R}^{\frac{1}{2}}]_{m,m}[\bold{S}\bold{B} ]_{i,j}
\\
&=-\frac{1}{L} \E \Tr \bold{R}\bold{Q} [ \bold{S} \bold{D} \bold{Z}^{H}\bold{Q}\bold{Z}\bold{B}]_{i,j}
+\frac{1}{L} \E \Tr \bold{R}\bold{Q} [ \bold{S}\bold{B} ]_{i,j}.
\end{aligned}
\end{equation}
Therefore, by moving the first term of the RHS to the LHS, we have
\begin{equation}
\begin{aligned}
 &\frac{1}{M}\E\Tr\bold{Z}^{H}\bold{Q}\bold{Z}\bold{B}
 \\
 &=\frac{1}{M}\E\Tr((\frac{1}{L}\E\Tr\bold{R}\bold{Q})^{-1}\bold{I}_{L}+ \bold{S} \bold{D} )^{-1} \bold{S}\bold{B}+\varepsilon_{z},
\end{aligned}
\end{equation}
where 
\begin{equation}
\begin{aligned}
|\varepsilon_{z}|&=|\frac{1}{M}\E \Tr((\frac{1}{L}\Tr\E \bold{R}\bold{Q})^{-1}\bold{I}_{L}+ \bold{S} \bold{D} )^{-1}
\\
&\times
\bold{S} \bold{D} \bold{Z}^{H}\bold{Q}\bold{Z}\bold{B} (\frac{1}{L} \Tr (\bold{R}\bold{Q})-\frac{1}{L}\E\Tr(\bold{R}\bold{Q}))| 
\\
& \le  \frac{K}{ML}\Var^{\frac{1}{2}}(\Tr(\bold{R}\bold{Q}))
=O(\frac{\| \bold{D} \|^2}{M}),
\end{aligned}
\end{equation}
with $K$ being a constant. The conclusion follows from the variance control in Appendix B of~\cite{rubio2012clt}.
\end{proof}

\section{Complexity Analysis for Solving~(\ref{basic_eq})}
\label{comp_ana}
To evaluate the complexity of~Algorithm~\ref{gra_alg}, we first investigate the complexity to obtain an $\varepsilon$-approximation of the solution for the canonical equation~(\ref{basic_eq}). Starting from a simpler case, we first investigate the $\varepsilon$-approximation for $g$ and $\overline{g}$ given $\delta$, which is given as the following lemma.
\begin{lemma}
\label{fg_lem}
Given $z$, consider the function $f(\overline{g})=\frac{1}{M}\Tr\bold{T}(\bold{I}+g\bold{T})^{-1} $ with $g=\frac{1}{M}\Tr\bold{S}(z\bold{I}+\overline{g}\bold{S})^{-1}$. An $\varepsilon$-solution for the equation $\overline{g}=f(\overline{g})$ can be given in $O(\log(\frac{1}{\varepsilon}))$ iterations.
\end{lemma}

\begin{proof} We define the iteration $\overline{g}^{(t+1)}=f(\overline{g}^{(t)})$.
First, we have the following bounds for $g$ and $\overline{g}$
\begin{equation}
\begin{aligned}
\frac{\Tr\bold{S}}{M(z+s_{max}t_{max} )}<g<\frac{Ls_{max}}{Mz}
\\
\frac{\Tr\bold{T}}{M(1+ \frac{Ls_{max}t_{max}}{Nz} )}<\overline{g}<t_{max}.
\end{aligned}
\end{equation}
which indicates that if we can obtain an $\varepsilon$-approximation for $g$, we could obtain an $\varepsilon$-approximation for $\overline{g}$ and vice versa.
$f(\overline{g})$ is monotonically increasing since
\begin{equation}
\label{fg_de}
f'(\overline{g})=\frac{\Tr\bold{T}^2(\bold{I}+g\bold{T})^{-2}}{M}\frac{\Tr\bold{S}^{2}(z\bold{I}+\overline{g}\bold{S})^{-2} }{M}>0.
\end{equation}
Meanwhile, it is easy to verify that $\frac{f(\overline{g})}{\overline{g}}$ is monotonically decreasing so we have $\frac{f(\overline{g})}{\overline{g}}<1$ when $\overline{g}>\overline{g}^{*}$. Therefore, if we start from $\overline{g}^{(0)}=t_{max}$, we have $\overline{g}^{(t+1)}<...<\overline{g}^{(1)}<\overline{g}^{(0)}$, indicating that $\overline{g}$ will converge to the true solution $\overline{g}^{*}$ decreasingly.
Meanwhile, we have
\begin{equation}
\begin{aligned}
&1=\frac{f(\overline{g})}{f(\overline{g})}=\frac{\frac{\Tr\bold{T}(\bold{I}+g\bold{T})^{-2}}{M}+\frac{g\Tr\bold{T}^2(\bold{I}+g\bold{T})^{-2}}{M} }{f(\overline{g})}
\\
&>\frac{(\frac{z\Tr\bold{S}(z\bold{I}+\overline{g}\bold{S})^{-2} }{M}+\frac{\overline{g}\Tr\bold{S}^2(z\bold{I}+\overline{g}\bold{S})^{-2} }{M})\frac{\Tr\bold{T}^2(\bold{I}+\overline{g}\bold{T})^{-2}}{M}}{f(\overline{g})}
  .
\end{aligned}
\end{equation}
By~(\ref{fg_de}), we have
\begin{equation}
\begin{aligned}
&0<f'(\overline{g})<\frac{f(\overline{g})}{\overline{g}}-\frac{z\Tr\bold{S}(z\bold{I}+\overline{g}\bold{S})^{-2} }{\overline{g} M}
\frac{\Tr\bold{T}^2(\bold{I}+\overline{g}\bold{T})^{-2}}{M}
\\
&<1 -\frac{z\Tr\bold{S}(z\bold{I}+\overline{g}\bold{S})^{-2} }{\overline{g} M}
\frac{\Tr\bold{T}^2(\bold{I}+\overline{g}\bold{T})^{-2}}{M}:=\beta_{g}<1.
\end{aligned}
\end{equation}
Therefore, by the mean value theorem, we have
\begin{equation}
\begin{aligned}
&|\overline{g}^{(t+1)}-\overline{g}^{*}| =| f(\overline{g}^{(t)})-f(\overline{g}^{*})|=|g'(\psi)(\overline{g}^{(t)}-\overline{g}^{*})|
\\
&<\beta_{g}|\overline{g}^{(t)}-\overline{g}^{*}|<\beta_{g}^{t+1}|\overline{g}^{(0)}|  ,
\end{aligned}
\end{equation}
which indicates that an $\varepsilon$-solution can be obtained in $O(\log(\frac{1}{\varepsilon}))$ iterations.
\end{proof}

Now we turn to investigate the $\varepsilon$-approximation for $\delta$, which is given in the following lemma.
\begin{lemma}
\label{h_del}
Consider the function $h(\delta)=\frac{1}{L}\Tr\bold{R}(z\bold{I}+\frac{g\overline{g}}{\delta}\bold{R})^{-1} $ with $\overline{g}=\frac{1}{M}\Tr\bold{T}(\bold{I}+g\bold{T})^{-1} $ and $g=\frac{1}{M}\Tr\bold{S}(\frac{1}{\delta}\bold{I}+\overline{g}\bold{S})^{-1}$. An $\varepsilon$-solution for the equation $\delta=h(\delta)$ can be given in $O(\log(\frac{1}{\varepsilon}))$ iterations.
\end{lemma}
\begin{proof} The proof is similar to that of Lemma~\ref{fg_lem}. First, we can obtain the bounds
\begin{equation}
\begin{aligned}
 &\delta_{L}= \frac{ \Tr\bold{R} }{L (z+ s_{max} t_{max} r_{max} ) }   \le  \delta \le \frac{N r_{max}}{Lz}=\delta_{U},
 \\
 &
g_{L}=\frac{\Tr\bold{S}}{  M(\frac{1}{\delta_{L} }+ s_{max}t_{max}) }     \le g \le \frac{N s_{max}r_{max} }{Mz}=g_{U},
 \\
 &
 \overline{g}_{L}= \frac{\Tr\bold{T}}{M(1+ \frac{N r_{max}t_{max}}{Lz})}   \le \overline{g} \le t_{max}=\overline{g}_{U}.
 \\
\end{aligned}
\end{equation}
Similar to the proof of Lemma~\ref{fg_lem}, we can verify that the convergence process is 
\begin{equation}
\begin{aligned}
\delta^{(0)}>\delta^{(1)}>...>\delta^{(n)}>...>\delta^{*},
\end{aligned}
\end{equation}
which means that $\delta$ decreases and converges to the solution. Also, $\frac{h(\delta)}{\delta}<1$ holds true when $\delta \in (\delta^{*},\infty)$. Next, we will bound the derivative $h'(\delta)$. Since
\begin{equation}
\begin{aligned}
1=h(\delta)/h(\delta)=\frac{\delta}{f(\delta)}(\frac{z}{\delta}\gamma_{R,I}+\frac{Mg\overline{g} }{L\delta^2} \gamma_{R}),
\end{aligned}
\end{equation}
there holds true that
\begin{equation}
\begin{aligned}
h'(\delta)&=\frac{Mg\overline{g}\gamma_{R}}{L\delta^2}-\frac{M(g-\overline{g}\gamma_{S} )(\overline{g}-g\gamma_{T} )\gamma_{R}}{L\delta^2\Delta_{T}}
\\
&=\frac{h(\delta)}{\delta}-\frac{z\gamma_{R,I}}{\delta}-\frac{M(g-\overline{g}\gamma_{S} )(\overline{g}-g\gamma_{T} )\gamma_{R}}{L\delta^2\Delta_{T}}
\\
&<1- \frac{1 }{(1+ s_{max} t_{max} r_{max}/z )^2 }:=\beta_{\delta}<1.
\end{aligned}
\end{equation}
Therefore, we can show that
\begin{equation}
\begin{aligned}
&|\delta^{(t+1)}-\delta^{*}| =| h(\delta^{(t)})-h(\delta^{*})|=|h'(\psi)(\delta^{(t)}-\delta^{*})|
\\
&<\beta_{\delta}|\delta^{(t)}-\delta|
<\beta^{t}|\delta^{(0)} |  ,
\end{aligned}
\end{equation}
which indicate that an $\varepsilon$-solution can be obtained in $O(\log(\frac{1}{\varepsilon}))$ iterations.
\end{proof}
However, the conclusion in Lemma~\ref{h_del} holds true when the accurate solution of $f(\overline{g})=\overline{g}$ can be obtained in each iteration. In practice, only an approximation for $f(\overline{g})=\overline{g}$ can be obtained.
By Lemma~\ref{fg_lem}, given $\delta$, we can find a solution $\hat{g}(\delta)$ such that $g-\varepsilon_{inner}<\hat{g}<g$ and use $\hat{h}(\delta)$ to represent the value computed based on $\hat{g}$ and $\hat{\overline{g}}$. Then, we will evaluate the gap between $\hat{h}$ and ${h}$. First, we have
\begin{equation}
\begin{aligned}
&0<\hat{h}(\delta)-h(\delta)=\frac{M}{L^2}\Tr(\FR^2\left(z\bold{I}+\frac{M \hat{g}\hat{\overline{g}}}{L\delta} \right)^{-1} 
\\
&\times\left(z\bold{I}+\frac{M g\overline{g}}{L\delta} \right)^{-1})\frac{g\overline{g}-\hat{g}\hat{\overline{g}} }{\delta}
\\
&\le \frac{MNr_{max}^2}{z^2L^2\delta_{L}} (g-\hat{g}) < \frac{MN r_{max}^2\varepsilon_{inner}}{z^2L^2\delta_{L}}
:=\varepsilon_{f} 
\end{aligned}
\end{equation}
\begin{equation}
\begin{aligned}
&|h(\delta^{(t)})-\hat{h}(\hat{\delta}^{(t)})|
\le |h(\delta^{(t)})-h(\hat{\delta}^{(t)})|
\\
&+|f(\hat{\delta}^{(t)})-\hat{f}(\hat{\delta}^{(t)})|
\le \beta_{\delta} | f(\delta^{(t-1)})-\hat{f}(\hat{\delta}^{(t-1)})| 
\\
&+ \varepsilon_{f}
\le \beta_{\delta}^2| f(\delta^{(t-2)})-\hat{f}(\hat{\delta}^{(t-2)}) |+\beta\varepsilon_{f}+\varepsilon_{f}
\\
&
\le ...
< \frac{\varepsilon_{f}}{1-\beta_{\delta}}.
\end{aligned}
\end{equation}
\begin{equation}
\begin{aligned}
&
|\hat{\delta}^{(t+1)}-\delta^{*}|=|\hat{f}(\hat{\delta}^{(t)})-f(\delta^{*})| \le |\hat{f}(\hat{\delta}^{(t)})-{f}({\delta}^{(t)})|
\\
&
+|\hat{f}({\delta}^{(t)})-{f}({\delta})|
< \frac{\varepsilon_{f}}{1-\beta_{\delta}}+\beta^{t}\delta^{(0)} 
\end{aligned}
\end{equation}
If we let $\varepsilon_{inner}=\frac{z^2L^2\delta_{L}(1-\beta_{\delta}) \varepsilon } {2MN r_{max}^2}$, there holds true that $ \frac{\varepsilon_{f}}{1-\beta}<\frac{\varepsilon}{2}$. Meanwhile, by letting $t\ge\lceil \frac{\log(\frac{\varepsilon}{2\delta^{(0)}})}{\log(\beta_{\delta})}\rceil$, we have $|\hat{\delta}^{(t+1)}-\delta^{*}| \le \varepsilon$.
Therefore, the complexity of obtaining an $\varepsilon$-approximation of $\delta^{*}$ is $O(N\log^2(\frac{1}{\varepsilon}))$, where $N$ comes from the calculation of the trace. The algorithm is given in Algorithm~\ref{alg_sol_10} , where $N_{1,max}$ and $N_{2,max}$ represent the number of iterations to obtain the $\varepsilon$ accuracy, which can be obtained by the above discussion. 
\begin{algorithm}
\caption{Algorithm for obtaining the $\varepsilon$-solution of the Canonical Equation~(\ref{basic_eq}).} 
\label{alg_sol_10} 
\begin{algorithmic}[1] 
\REQUIRE $z$, $\FR$, $\FS$, $\FT$, $N_{1,max}$, $N_{2, max}$.
\STATE $\delta^{(0)}>\delta^{U}$, $t_1=0$.
\REPEAT
\STATE Set $\overline{g}^{(0)}=\overline{g}_{U}$, $t_{2}=1$
\REPEAT
\STATE $g^{(t_{2})}=\frac{1}{M}\Tr\bold{S}(\frac{1}{\delta^{(t_1)}}\bold{I}+\overline{g}^{(t_2-1)}\bold{S})^{-1}$
\STATE $\overline{g}^{(t_{2})}=\frac{1}{M}\Tr\bold{T}(\bold{I}+{g}^{(t_2)}\bold{T})^{-1}$
\STATE $t_2=t_2+1$.
\UNTIL $t_2 > N_{2,max}$.
\STATE $\delta^{t_1}=\frac{1}{L}\Tr\bold{R}(z\bold{I}+\frac{Mg^{(t_2)}\overline{g}^{(t_2)}}{L\delta^{(t_1-1)} }\FR  )^{-1}$.
\STATE $t_1=t_1+1$.
\UNTIL $t_1>N_{1,max}$.
\ENSURE  $\delta$, $g$, $\overline{g}$.
\end{algorithmic}
\end{algorithm}

\ifCLASSOPTIONcaptionsoff
  \newpage
\fi

\bibliographystyle{IEEEtran}
\bibliography{IEEEabrv,ref}

\end{document}